\documentclass[a4paper,reqno]{amsart}

\usepackage{graphicx}
\usepackage{xcolor}
\usepackage{tikz}
\usepackage{caption}
\usepackage{booktabs}
\usepackage{adjustbox}
\usepackage{longtable}
\usepackage{array}
\newcommand{\be}{\begin{equation}}
\newcommand{\ee}{\end{equation}}
\newcommand{\dis}{\displaystyle}

\usepackage{enumerate}
\usepackage{enumitem}

\newtheorem{thm}{Theorem}[section]
\newtheorem{prop}[thm]{Proposition}
\newtheorem{theorem}[thm]{Theorem}

\newtheorem{lemma}[thm]{Lemma}
\newtheorem{dfn}[thm]{Definition}
\newtheorem{remark}[thm]{\it Remark}
\newtheorem{example}[thm]{\it Example}

\usepackage{tikz-3dplot}
\usepackage{xifthen}
\usepackage{caption}

\tdplotsetmaincoords{60}{125}
\tdplotsetrotatedcoords{0}{20}{0} 

\usepackage[
labelfont=sf,
hypcap=false,
format=hang
]{caption}

\begin{document}

\title{Integrable multi-component difference systems of  equations}

\author{Pavlos Kassotakis }
\address{Department of Mathematics and Statistics, University of Cyprus, P.O Box: 20537, 1678 Nicosia, Cyprus;}
\curraddr{}
\email{pavlos1978@gmail.com, pkasso01@ucy.ac.cy}
\thanks{}

\author{Maciej Nieszporski }
\address{Katedra Metod Matematycznych Fizyki, Wydzia\l{} Fizyki, Uniwersytet Warszawski,
ul. Pasteura 5, 02-093 Warszawa, Poland;}
\curraddr{}
\email{maciejun@fuw.edu.pl}
\thanks{}

\author{Vassilios  Papageorgiou }
\address{Department of Mathematics, University of Patras, 26 500 Patras, Greece;}
\curraddr{}
\email{vassilis@math.upatras.gr}
\thanks{}
\author{Anastasios Tongas}
\address{Department of Mathematics, University of Patras, 26 500 Patras, Greece;}
\curraddr{}
\email{tasos@math.upatras.gr}
\thanks{}

\maketitle
\begin{abstract}
We present two lists 
of multi-component systems of integrable difference equations defined on the edges of a $\mathbb{Z}^2$ graph. The integrability of these  systems is manifested by their Lax formulation which is a consequence of the multi-dimensional compatibility  of these systems. Imposing constraints consistent with the systems of difference equations, we recover known integrable quad-equations including the discrete version of the Krichever-Novikov  equation.
 The systems of difference equations allow us for a straightforward reformulation as Yang-Baxter maps. Certain two-component  systems of equation defined on the vertices of a $\mathbb{Z}^2$ lattice, their non-potential form and  integrable equations defined on 5-point stencils, are also obtained.

 \begin{center}
\emph{Dedicated to Allan Fordy on the occasion of his 70th birthday.}
    \end{center}
\end{abstract}

\setcounter{tocdepth}{2}

\tableofcontents


\section{Introduction}

Among the nonlinear difference equations, of particular interest are the ones that admit Lax pair formulation,
B\"acklund transformation, nonlinear superposition principle and other mutually related features that allow one to call them integrable \cite{Franks-book}.

Nonlinear integrable discrete equations, most often are defined on vertices of ${\mathbb Z}^2$ and ${\mathbb Z}^3$ lattices. For instance in \cite{ABS} one can find a list of such  integrable equations defined on the vertices of elementary quadrangle  of a $\mathbb{Z}^2$ lattice.
Using the shorthand notation $x_{12}:=x(m+1,n+1)$, $x_{1}:=x(m+1,n)$, $x_{2}:=x(m,n+1)$, $x:=x(m,n)$
the list can be written as
\begin{align*}
\begin{aligned}
p(xx_1+x_2x_{12})-q(xx_2+x_1x_{12})-\frac{pQ-Pq}{1-p^2q^2}\left(xx_{12}+x_1x_2-pq(1+xx_1x_2x_{12})\right)=0,
\end{aligned}&& (Q4)\\
\begin{aligned}
\left(p-\frac{1}{p}\right)(xx_1+x_2x_{12})-\left(q-\frac{1}{q}\right)(xx_2+x_1x_{12})-\left(\frac{p}{q}-\frac{q}{p}\right)\left(x_1x_2+xx_{12}\right)\\
-\delta\left(\frac{p}{q}-\frac{q}{p}\right)\left(p-\frac{1}{p}\right)\left(q-\frac{1}{q}\right)=0,
\end{aligned}&& (Q3^{\delta})\\
\begin{aligned}
p(x-x_2)(x_1-x_{12})-q(x-x_1)(x_2-x_{12})+pq(p-q)(x+x_1+x_2+x_{12}-p^2+pq-q^2)=0,
\end{aligned}&& (Q2)\\
\begin{aligned}
p(x-x_2)(x_1-x_{12})-q(x-x_1)(x_2-x_{12})+\delta pq(p-q)=0,
\end{aligned}&& (Q1^{\delta})\\
\begin{aligned}
\left(p-\frac{1}{p}\right)(xx_2+x_1x_{12})-\left(q-\frac{1}{q}\right)(xx_1+x_2x_{12})-\left(\frac{p}{q}-\frac{q}{p}\right)(1+xx_1x_2x_{12})=0,
\end{aligned}&& (A2)\\
\begin{aligned}
p(x+x_2)(x_1+x_{12})-q(x+x_1)(x_2+x_{12})+\delta pq(p-q)=0,
\end{aligned}&& (A1^{\delta})\\
\begin{aligned}
p(xx_1+x_2x_{12})-q(xx_2+x_1x_{12})+\delta (p^2-q^2)=0,
\end{aligned}&& (H3^{\delta})\\
\begin{aligned}
(x-x_{12})(x_1-x_2)-(p-q)(x+x_1+x_2+x_{12}+p+q)=0,
\end{aligned}&& (H2)\\
\begin{aligned}
(x-x_{12})(x_1-x_2)+p-q=0,
\end{aligned}&& (H1)
\end{align*}
where $\delta=0,1,$ $p$ is prescribed function of $m$, $q$ is prescribed function of $n$ and $P^2=p^4-\gamma p^2+1$, $Q^2=q^4-\gamma q^2+1$, so $(p,P)$ and $(q,Q)$ are points on the elliptic curve $\mathcal{E}=\{(\alpha,A)\in\mathbb{C}^2:A^2=\alpha^4-\gamma\alpha^2+1\},$ where $\gamma$ the modulus of $\mathcal{E}$. The top equation, $Q4$ serves as the discrete version of Krichever-Novikov equation and it was obtained in \cite{Adler:1998}. Here we have presented $Q4$ in its Jacobi form that was introduced in \cite{Hietarinta:Q4}.
Equations $Q1^{\delta}$ and $A1^{\delta}$ are  equivalent under a point transformations, so do $A2$ and $Q3^{0}$.


The six bottom members of the list, namely $(Q1^{\delta},A2,A1^{\delta},H3^{\delta},H2,H1)$ admit a Lie point symmetry group. The existence of the Lie point symmetry group allows one to reformulate these six  members of the list as systems of difference equations defined on the edges of a quadrangle  of the $\mathbb{Z}^2$ graph \cite{Tasos}, we refer to these systems shortly as {\it bond systems} \cite{hi-via,KaNie,KaNie1,KaNie3,KaNie:2018}. The hallmark of integrability of the bond systems is the 3-dimensional compatibility in analogy with the consistency-around-the-cube \cite{FN} of the corresponding quad-equations. In turn, the 3-dimensional compatibility of a bond system, implies $n-$dimensional (multi-dimensional) compatibility of the same system \cite{ABSf}.

The top three equations of the list, namely  $Q4,$ $Q3^{\delta}$ and $Q2$ do not admit any Lie-point symmetry.
In this manuscript, we propose a procedure that allows us to obtain from any member of the list   a multi-dimensional compatible (see Section \ref{Section:2})  multi-component system of difference equations defined on edges of a $\mathbb{Z}^2$ graph. Moreover, we  associate a multi-field {\it quadrirational}\footnote{The notion of quadrirational maps was introduced in \cite{et-2003,ABSf}. This notion can be extended to $n-$dimensions to the notion of {\it $2^n-$rational maps}  \cite{2n-rat}. } Yang-Baxter map with each one member of the bond systems. Also, we recover the original list  of equations $(Q4,\ldots, H1)$ from certain  reductions of these novel systems.

 It is worth mentioning that in  \cite{Tasos:ell} systems of difference equations  were associated with all the members of the list but these systems turned out to be too restrictive i.e. there was a restrictive set of initial values for which the solution could emerge. Moreover, the Yang-Baxter maps associated with the systems introduced in \cite{Tasos:ell}, turned out to be non-quadrirational as a consequence of the restrictive initial value problem. The paper is organized as follows.

In Section \ref{Section:2} we set the notation we are using through out the manuscript.  In Section \ref{Section:3} we present as a motivating example, the lattice potential KdV as a difference system in bond variables.  Two lists, of seven members each, of integrable difference systems in bond variables, are presented in Section \ref{Section:4}. We refer to these lists as the $m-$list and the $a-$list, respectively, and by imposing constraints which are consistent with the members of the lists, we obtain as reductions well known integrable quad-equations, including the discrete version of the Krichever-Novikov equation. The discrete zero curvature representation of all members of the $m-$list and the $a-$list is obtained  in Section \ref{Section:5}. In Section \ref{Section:6}, for some members of the $m-$list and the $a-$list, we obtain the associated systems of vertex equations, which are two-component integrable  systems defined on vertices of a $\mathbb{Z}^2$ lattice. Moreover, we obtain the non-potential form of these vertex systems of equations together with some integrable equations defined on 5-point stencils. The final Section 7 contains a brief discussion of the results obtained in this paper and some open questions that could be addressed for a future study. There are three appendices to this paper regarding the reinterpretation of the derived systems as Yang-Baxter maps. Appendices A and B contain the quadrirational Yang-Baxter maps obtained from the $m-$list and a-list, respectively. Appendix C contains a proof and a technical result of the multi-dimensional consistency property, using the $n$-factorization property satisfying the Lax matrices of the corresponding Yang-Baxter maps.


\section{Notation} \label{Section:2}

We consider the following sets which are associated with the $\mathbb{Z}^2$ lattice (see Figure \ref{fig00}):
\begin{itemize}
\item the set of vertices $V=\left\{(m,n) \, | \,  m,n \in {\mathbb Z} \right\},$  which is the original ${\mathbb Z}^2$ lattice,
\item the set of horizontal edges $E_h= \left\{ \{(m,n),(m+1,n)\} \, | \,  m,n \in {\mathbb Z} \right\},$
\item the set of vertical edges  $E_v= \left\{ \{(m,n),(m,n+1)\} \, | \,  m,n \in {\mathbb Z} \right\}.$
\end{itemize}
\begin{figure}[h]
\begin{minipage}[h]{0.3\textwidth}
\begin{tikzpicture}
\draw[dashed] (-0.5,0)--(2.5,0);
\draw[dashed] (-0.5,1)--(2.5,1);
\draw[dashed] (-0.5,2)--(2.5,2);
\draw[dashed] (2,-0.5)--(2,2.5);
\draw[dashed] (1,-0.5)--(1,2.5);
\draw[dashed] (0,-0.5)--(0,2.5);
\filldraw
(0,0) circle (2pt) (1,0) circle (2pt) (2,0) circle (2pt) (0,1) circle (2pt) (0,2) circle (2pt) (1,2) circle (2pt) (1,1) circle (2pt) (2,1) circle (2pt) (2,2) circle (2pt);
\end{tikzpicture}
\end{minipage}
\begin{minipage}[h]{0.3\textwidth}
\begin{tikzpicture}
\draw[dashed] (-0.5,0)--(2.5,0);
\draw[dashed] (-0.5,1)--(2.5,1);
\draw[dashed] (-0.5,2)--(2.5,2);
\draw[dashed] (2,-0.5)--(2,2.5);
\draw[dashed] (1,-0.5)--(1,2.5);
\draw[dashed] (0,-0.5)--(0,2.5);
\draw[ultra thick] (-0.5,0)--(-0.1,0);
\draw[ultra thick] (0.1,0)--(0.9,0);
\draw[ultra thick] (1.1,0)--(1.9,0);
\draw[ultra thick] (2.1,0)--(2.5,0);
\draw[ultra thick] (-0.5,1)--(-0.1,1);
\draw[ultra thick] (0.1,1)--(0.9,1);
\draw[ultra thick] (1.1,1)--(1.9,1);
\draw[ultra thick] (2.1,1)--(2.5,1);
\draw[ultra thick] (-0.5,2)--(-0.1,2);
\draw[ultra thick] (0.1,2)--(0.9,2);
\draw[ultra thick] (1.1,2)--(1.9,2);
\draw[ultra thick] (2.1,2)--(2.5,2);
\end{tikzpicture}
\end{minipage}
\begin{minipage}[h]{0.3\textwidth}
\begin{tikzpicture}
\draw[dashed] (-0.5,0)--(2.5,0);
\draw[dashed] (-0.5,1)--(2.5,1);
\draw[dashed] (-0.5,2)--(2.5,2);
\draw[dashed] (2,-0.5)--(2,2.5);
\draw[dashed] (1,-0.5)--(1,2.5);
\draw[dashed] (0,-0.5)--(0,2.5);
\draw[ultra thick] (0,-0.5)--(0,-0.1);
\draw[ultra thick] (0,0.1)--(0,0.9);
\draw[ultra thick] (0,1.1)--(0,1.9);
\draw[ultra thick] (0,2.1)--(0,2.5);
\draw[ultra thick] (1,-0.5)--(1,-0.1);
\draw[ultra thick] (1,0.1)--(1,0.9);
\draw[ultra thick] (1,1.1)--(1,1.9);
\draw[ultra thick] (1,2.1)--(1,2.5);
\draw[ultra thick] (2,-0.5)--(2,-0.1);
\draw[ultra thick] (2,0.1)--(2,0.9);
\draw[ultra thick] (2,1.1)--(2,1.9);
\draw[ultra thick] (2,2.1)--(2,2.5);
\end{tikzpicture}
\end{minipage}
\caption{Sets associated with the ${\mathbb Z}^2-$lattice}
\label{fig00}
\end{figure}
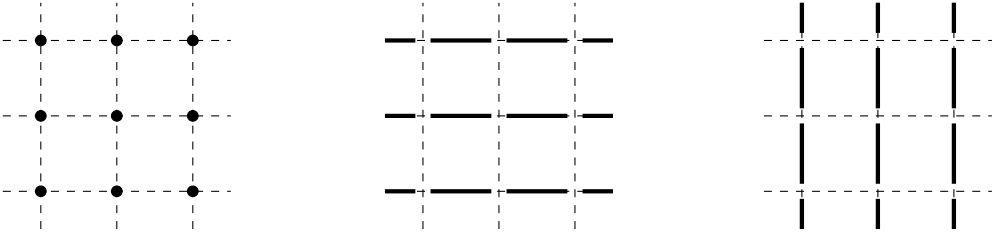
The set $V\cup E_h\cup E_v$ constitutes the $\mathbb{Z}^2$ graph.

 We label with $(m,n) \in {\mathbb Z}^2$,
$(m_1,n_1) \in {\mathbb Z}^2$ and $(m_2,n_2) \in {\mathbb Z}^2$   the elements of $V$, $E_h$ and $E_v$  respectively. Therefore, each of this sets can be regarded as a ${\mathbb Z}^2$ lattice itself, where we use the convention  to denote the edge $\{(0,0),(1,0)\}$ as $(m_1,n_1)=(0,0),$ the edge $\{(0,0),(0,1)\}$ as $(m_2,n_2)=(0,0)$.

We consider the functions $x,z: V \rightarrow \mathbb{C},$ as well as  $u,s: E_h\rightarrow \mathbb{C}$ and finally $v,t: E_v\rightarrow \mathbb{C}.$ Therefore,
$(m,n) \mapsto x_{m,n},$  $(m,n) \mapsto z_{m,n},$ $(m_1,n_1) \mapsto u_{m_1,n_1},$ $(m_1,n_1) \mapsto s_{m_1,n_1},$ $(m_2,n_2) \mapsto v_{m_2,n_2}$ and finally $(m_2,n_2) \mapsto t_{m_2,n_2}$. We also consider the functions $p: E_h\ni(m_1,n_1)\mapsto p_{m_1}\in\mathbb{C}$ and $q: E_v\ni(m_2,n_2)\mapsto q_{n_2}\in\mathbb{C}$. Moreover, we introduce a concise notation, first, by considering  functions $X^i, Y^i, p^i, i=1,2$ such that
 $X^1_{m_1,n_1}:=u_{m_1,n_1},$ $X^2_{m_2,n_2}:=v_{m_2,n_2},$ $p^1_{m_1}:=p_{m_1},$  $Y^1_{m_1,n_1}:=s_{m_1,n_1},$ $Y^2_{m_2,n_2}:=t_{m_2,n_2},$ and $p^2_{n_2}:=q_{n_2}$, second,  we omit the dependency of the functions $X^i, Y^i, p^i, i=1,2$ on the independent variables, third, we denote the shifts in the corresponding variables as subscripts i.e. $X^i_1=X^i_{m_i+1,n_i}, X^i_2=X^i_{m_i,n_i+1}, X^i_{11}=X^i_{m_i+2,n_i},$ $X^i_{-1}=X^i_{m_i-1,n_i},$ $X^i_{-2}=X^i_{m_i,n_i-1},$  $X^i_{-12}=X^i_{m_i-1,n_i+1}$ etc.
 Therefore with  this concise notation we enumerate functions with superscripts and the shifts with subscripts, see Figure \ref{notation1}.
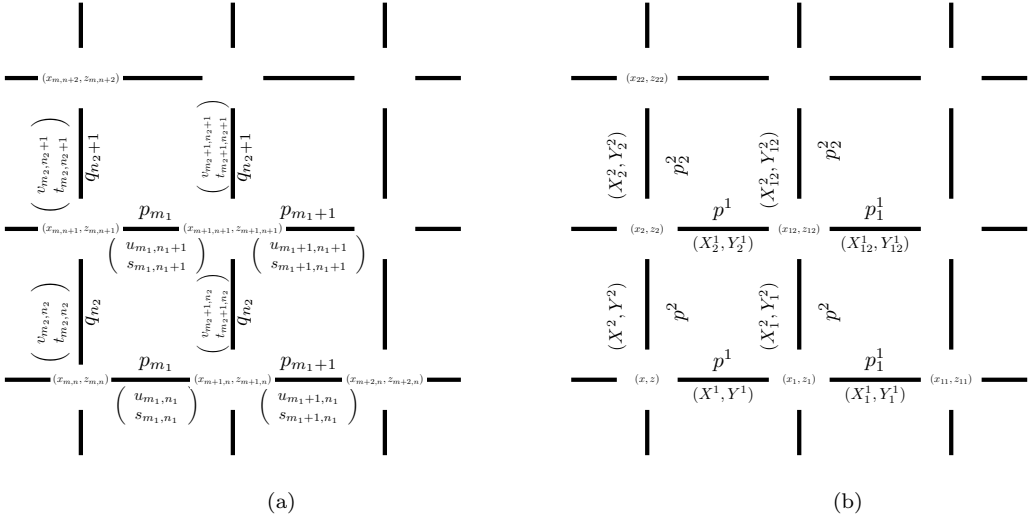
\begin{figure}[h]
\begin{minipage}[h]{0.45\textwidth}
\begin{tikzpicture}[fill=white, scale=2, every node/.style={transform shape}];
\draw[ultra thick,fill=white] (-0.5,0)--(-0.2,0)  (0.5,0) node [scale=0.3, below ] { $\left(\begin{array}{c}
                                                                                 u_{m_1,n_1}\\
                                                                                 s_{m_1,n_1}
                                                                                 \end{array}\right)$}     (0.5,0) node [scale=0.4, above ] { $ p_{m_1}$};

\draw[ultra thick,fill=white] (0.2,0)--(0.8,0);
\draw[ultra thick,fill=white] (1.2,0)--(1.8,0)    (1.5,0) node [scale=0.3, below ] { $ \left(\begin{array}{c}
                                                                                 u_{m_1+1,n_1}\\
                                                                                 s_{m_1+1,n_1}
                                                                                 \end{array}\right)$} (1.5,0) node [scale=0.4, above ] { $ p_{m_1+1}$};
\draw[ultra thick,fill=white] (2.2,0)--(2.5,0);
\draw[ultra thick,fill=white] (-0.5,1)--(-0.2,1)  (0.5,1) node [scale=0.3, below ] { $ \left(\begin{array}{c}
                                                                                 u_{m_1,n_1+1}\\
                                                                                 s_{m_1,n_1+1}
                                                                                 \end{array}\right)$} (0.5,1) node [scale=0.4, above ] { $ p_{m_1}$};
\draw[ultra thick,fill=white] (0.2,1)--(0.8,1);
\draw[ultra thick,fill=white] (1.2,1)--(1.8,1)    (1.5,1) node [scale=0.3, below ] { $ \left(\begin{array}{c}
                                                                                 u_{m_1+1,n_1+1}\\
                                                                                 s_{m_1+1,n_1+1}
                                                                                 \end{array}\right)$}  (1.5,1) node [scale=0.4, above ] { $ p_{m_1+1}$};;
\draw[ultra thick,fill=white] (2.2,1)--(2.5,1);
\draw[ultra thick,fill=white] (-0.5,2)--(-0.2,2);
\draw[ultra thick] (0.2,2)--(0.8,2);
\draw[ultra thick] (1.2,2)--(1.8,2);
\draw[ultra thick] (2.2,2)--(2.5,2);
\draw  (0,0) node [scale=0.2 ] { $ (x_{m,n},z_{m,n})$};
\draw    (1,0) node [ fill=white, scale=0.2] { $(x_{m+1,n},z_{m+1,n})$};
\draw    (2,0) node [ fill=white, scale=0.2]  { $(x_{m+2,n},z_{m+2,n})$};
\draw    (0,1) node [ fill=white, scale=0.2 ] {$(x_{m,n+1},z_{m,n+1})$};
\draw    (0,2) node [fill=white, scale=0.2 ] {  $(x_{m,n+2},z_{m,n+2})$};
\draw    (1,1) node [ fill=white, scale=0.2 ] { $(x_{m+1,n+1},z_{m+1,n+1})$};

\draw[ultra thick] (0,-0.5)--(0,-0.2)   (-0.2,0.7) node [scale=0.3, left, rotate=90 ] {$\left(\begin{array}{c}
                                                                                 v_{m_2,n_2}\\
                                                                                 t_{m_2,n_2}
                                                                                 \end{array}\right)$}  (0.87,0.75) node [scale=0.25, left, rotate=90 ] {$\left(\begin{array}{c}
                                                                                 v_{m_2+1,n_2}\\
                                                                                 t_{m_2+1,n_2}
                                                                                 \end{array}\right)$} (0.08,0.3) node [scale=0.4, right, rotate=90 ] {$q_{n_2}$} (1.08,0.3) node [scale=0.4, right, rotate=90 ] {$q_{n_2}$} (0.08,1.25) node [scale=0.4, right, rotate=90 ] {$q_{n_2+1}$} (1.08,1.25) node [scale=0.4, right, rotate=90 ] {$q_{n_2+1}$};
\draw[ultra thick] (0,0.1)--(0,0.8)  (-0.2,1.8) node [scale=0.3, left, rotate=90 ] {$\left(\begin{array}{c}
                                                                                 v_{m_2,n_2+1}\\
                                                                                 t_{m_2,n_2+1}
                                                                                 \end{array}\right)$}  (0.87,1.9) node [scale=0.25, left, rotate=90 ] {$\left(\begin{array}{c}
                                                                                 v_{m_2+1,n_2+1}\\
                                                                                 t_{m_2+1,n_2+1}
                                                                                 \end{array}\right)$} ;
\draw[ultra thick] (0,1.2)--(0,1.8);
\draw[ultra thick] (0,2.2)--(0,2.5);
\draw[ultra thick] (1,-0.5)--(1,-0.2);
\draw[ultra thick] (1,0.2)--(1,0.8);
\draw[ultra thick] (1,1.2)--(1,1.8);
\draw[ultra thick] (1,2.2)--(1,2.5);
\draw[ultra thick] (2,-0.5)--(2,-0.2);
\draw[ultra thick] (2,0.2)--(2,0.8);
\draw[ultra thick] (2,1.2)--(2,1.8);
\draw[ultra thick] (2,2.2)--(2,2.5);
%
\end{tikzpicture}
\captionsetup{font=footnotesize}
\captionof*{figure}{(a)}
\end{minipage}
\begin{minipage}[h]{0.45\textwidth}
\begin{tikzpicture}[scale=2, every node/.style={transform shape}];
\draw  (0,0) node [scale=0.2 ] { $ (x,z)$};
\draw    (1,0) node [scale=0.2 ] { $(x_1,z_1)$};
\draw    (2,0) node [scale=0.2]  { $(x_{11},z_{11})$};
\draw    (0,1) node [scale=0.2 ] {$(x_2,z_2)$};
\draw    (0,2) node [scale=0.2 ] {  $(x_{22},z_{22})$};
\draw    (1,1) node [scale=0.2 ] { $(x_{12},z_{12})$};
\draw[ultra thick] (-0.5,0)--(-0.2,0)  (0.5,0) node [scale=0.3, below ] { $ (X^1,Y^1)$} node [scale=0.4, above ] { $ p^1$};
\draw[ultra thick] (0.2,0)--(0.8,0);
\draw[ultra thick] (1.2,0)--(1.8,0)    (1.5,0) node [scale=0.3, below ] { $ (X^1_1,Y^1_1)$} node [scale=0.4, above ] { $ p^1_1$};
\draw[ultra thick] (2.2,0)--(2.5,0);
\draw[ultra thick] (-0.5,1)--(-0.2,1)  (0.5,1) node [scale=0.3, below ] { $ (X^1_2,Y^1_2)$} node [scale=0.4, above ] { $ p^1$} ;
\draw[ultra thick] (0.2,1)--(0.8,1);
\draw[ultra thick] (1.2,1)--(1.8,1)    (1.5,1) node [scale=0.3, below ] { $ (X^1_{12},Y^1_{12})$} node [scale=0.4, above ] { $ p^1_1$};
\draw[ultra thick] (2.2,1)--(2.5,1);
\draw[ultra thick] (-0.5,2)--(-0.2,2);
\draw[ultra thick] (0.2,2)--(0.8,2);
\draw[ultra thick] (1.2,2)--(1.8,2);
\draw[ultra thick] (2.2,2)--(2.5,2);
\draw[ultra thick] (0,-0.5)--(0,-0.2)   (-0.2,0.7) node [scale=0.34, left, rotate=90 ] {$(X^2,Y^2)$} (0.2,0.33) node [scale=0.4, right, rotate=90 ] {$p^2$}  (0.8,0.7) node [scale=0.34, left, rotate=90 ] {$(X^2_1,Y^2_1)$} (1.2,0.33) node [scale=0.4, right, rotate=90 ] {$p^2$} ;
\draw[ultra thick] (0,0.2)--(0,0.8)  (-0.2,1.7) node [scale=0.34, left, rotate=90 ] {$(X^2_2,Y^2_2)$} (0.2,1.3) node [scale=0.4, right, rotate=90 ] {$p^2_2$}  (0.8,1.7) node [scale=0.34, left, rotate=90 ] {$(X^2_{12},Y^2_{12})$} (1.2,1.4) node [scale=0.4, right, rotate=90 ] {$p^2_2$};
\draw[ultra thick] (0,1.2)--(0,1.8);
\draw[ultra thick] (0,2.2)--(0,2.5);
\draw[ultra thick] (1,-0.5)--(1,-0.2);
\draw[ultra thick] (1,0.2)--(1,0.8);
\draw[ultra thick] (1,1.2)--(1,1.8);
\draw[ultra thick] (1,2.2)--(1,2.5);
\draw[ultra thick] (2,-0.5)--(2,-0.2);
\draw[ultra thick] (2,0.2)--(2,0.8);
\draw[ultra thick] (2,1.2)--(2,1.8);
\draw[ultra thick] (2,2.2)--(2,2.5);
\end{tikzpicture}
\captionsetup{font=footnotesize}
\captionof*{figure}{(b)}
\end{minipage}
\caption{Fields on  the vertices and on the edges of the ${\mathbb Z}^2$ lattice. \newline (a): Standard notation. (b): concise notation used in the paper}\label{notation1}
\end{figure}

Having at hand the concise notation, we can now re-write the lattice potential KdV equation $(H1),$ in terms of the functions $(u,v,u_2,v_1)$ which are naturally defined on the edges of the $\mathbb{Z}^2$ graph.
So the lattice potential KdV
$$
(x_{12}-x)(x_1-x_2)=p-q,\quad (H1),
$$
in terms of the following  of functions
$$
\begin{array}{ll}
u:=x_1x,\;\;v:=x_2x,& \;\;u_2:=x_{12}x_2,\;\;v_1:=x_{12}x_1,
\end{array}
$$
 reads
\be  \label{F4-bs}
u_2=v\left(1+\frac{p-q}{u-v}\right),\quad v_1=u\left(1+\frac{p-q}{u-v}\right).
\ee
The functions $u,v,u_2,v_1$ are not independent since the relation $u u_2=v v_1$ holds.
After the identification $u=X^1, V=X^2, u_2=X^1_2, v_1=X^2_1, \; p=p^1,\; q=p^2,$ the bond system $(\ref{F4-bs})$ reads
$$
X^i_j=X^j\left(1+\frac{p^i-p^j}{X^i-X^j}\right),\;i\neq j\in\{1,2\}.
$$
The difference system in bond variables $(\ref{F4-bs}),$ was first considered in \cite{Tasos}.  The map $\phi: (u,v)\mapsto (u_2,v_1),$ where $u_2, v_1$ are given by (\ref{F4-bs}), is exactly the $F_{IV}$ quadrirational Yang-Baxter map \cite{ABSf}.  The following remarks are in order.
\begin{itemize}
\item The function $H=\frac{u}{v},$ is an alternating invariant  of the map $\phi$, i.e.  $H(u_2,v_1)=\frac{1}{H(u,v)}$ holds.
\item The system of difference equations $(\ref{F4-bs})$ together with the relation $u_2 u=v_1v,$ 
     leads to the $H1$ integrable lattice equation, for it guarantees the existence of a potential function $x$ such that $u=x_1x$ and $v=x_2x.$
\item    The system of difference equations $(\ref{F4-bs})$ together with the relation $u_2 u=v_1v,$ or  more surprisingly, even without this relation  is extendable to multi-dimensions i.e. the difference equations
    \be \label{f4-md}
    X^i_j=X^j\left(1+\frac{p^i-p^j}{X^i-X^j}\right),\;i\neq j\in\{1,2,\ldots, n\},
    \ee
    are compatible since $X^i_{jk}=X^i_{kj},\;i\neq j\neq k\in\{1,2,\ldots, n\}$ holds. The fact that $(\ref{F4-bs})$ is multi-dimensional compatible, even without the relation $u_2 u=v_1v,$ is an essential observation of this paper.
\end{itemize}

\section{ The lattice potential KdV equation  as a multi-component difference system in bond variables, a motivating example} \label{Section:3}


Motivated by the fact that $Q2, Q3^{\delta}$ and $Q4$ do not admit any Lie-point symmetry group, we investigate how $H1$ can be reformulated in terms of functions which are not invariants of any  Lie-point symmetry group of the latter. For example, in terms of $s:=\frac{x_1}{x},\; t:=\frac{x_2}{x},$ $H1$ reads:
\be \label{lpkdv-c}
x_2x(s_2-1/t)(s-t)=p-q,\quad x_1x(t_1-1/s)(s-t)=p-q.
\ee
We observe that the  functions $u=x_1x, \; v=x_2x,$ arise. Hence, in order to rewrite $H1$ in terms of the functions $s$ and $t,$  we are forced to consider the functions $u$ and $v$ and  we get the following system
\be  \label{F4-bs-c}
\begin{array}{l}
{\dis u_2=v\left(1+\frac{p-q}{u-v}\right),\quad v_1=u\left(1+\frac{p-q}{u-v}\right),} \\
{\dis s_2=\frac{1}{t}+\frac{p-q}{v(s-t)},\quad t_1=\frac{1}{s}+\frac{p-q}{u(s-t)}},
\end{array}
\ee
supplemented with the constraint $tu=sv$ which is  a consequence of the definition of the functions $u,s,v$ and $t$.
The following remarks are in order.
\begin{itemize}
\item  The relation $tu=sv,$ is {\it consistent} with the system of difference equations  $(\ref{F4-bs-c})$, i.e. if we impose a staircase initial value problem that satisfies the conditions $tu=sv,$ and $s_2u_2=t_1v_1$ on the staircase, then  due to equations $(\ref{F4-bs-c})$ the constraints hold on the whole lattice.
\item The system difference equations $(\ref{F4-bs-c})$ together with the relation $tu=sv,$ 
    leads to the $H1$ integrable lattice equation.
\item  The system difference equations $(\ref{F4-bs-c})$  without the relation $tu=sv,$  is not   multi-dimensional compatible.
\end{itemize}
In order to achieve multidimensional compatibility, we use the relation $tu=sv$ and we re-write (\ref{F4-bs-c}) s.t. $u_2$ to be a function of the variables $s,v,t$ only, $v_1$ to be a function of  $u,s,t$ only, $s_2$ to be a function of $u,v,t$ only, and $t_1$ to be a function of $u,s,v$ only,  to obtain
$$
\begin{array}{ll}
\begin{array}{l}
{\dis u_2=v+t\frac{p-q}{s-t},\quad v_1=u+s\frac{p-q}{s-t},} \\
{\dis s_2=\frac{1}{t}+\frac{p-q}{t(u-v)},\quad t_1=\frac{1}{s}+\frac{p-q}{s(u-v)}}.
\end{array}& (m\mathcal{H}1)
\end{array}
$$
$m\mathcal{H}1$ as it is shown in what follows,  is   multi-dimensional compatible as it stands (without using the relation $tu=sv$).

To recapitulate,
we consider the set of functions
$$
u:=x_1 x,\quad v:=x_2 x, \quad s:=x_1/x,\quad t:=x_2/x,
$$
as well as
$$
u_2:=x_{12} x_2,\quad v_1:=x_{12} x_1, \quad s_2:=x_{12}/x_2,\quad t_1:=x_{12}/x_1.
$$
 $H1$ can be written in a unique way at a time in terms of $(u_2,v,s,t),$ of $(v_1,u,s,t),$ of $(s_2,u,v,t)$ and in terms of $(t_1,u,v,s),$ respectively and leads exactly to $m\mathcal{H}1$.

\begin{prop}
The system of difference equations $m\mathcal{H}1$:
\begin{enumerate}
\item satisfies: 
\begin{itemize}
                                                 \item $\frac{s_2}{t_1}=\frac{s}{t},$ \\
                                                 \item  $u_2-v_1+\frac{p-q}{2}=-(u-v+\frac{p-q}{2})$
                                                 \end{itemize}
\item the relation $tu=sv$ holds if and only if $u_2s_2=t_1v_1$ holds
\item it  can be extended to multi-dimensions as follows:
\be \label{multi-f4}
X^i_j=X^j+Y^j\frac{p^i-p^j}{Y^i-Y^j},\quad Y^i_j=\frac{1}{Y^j}+\frac{1}{Y^j}\frac{p^i-p^j}{X^i-X^j},\;\;i\neq j\in \{1,\ldots,n\}.
\ee
where the compatibility conditions $X^i_{jk}=X^i_{kj}, \; Y^i_{jk}=Y^i_{kj}$ holds
\item it arises as the compatibility condition
\be \label{comex}
L({u_2,s_2};p,\lambda)\, L({v,t};q,\lambda) =L(v_1,t_1;q,\lambda)\, L({u,s};p,\lambda)
\ee
 of the linear system
$$
\Psi_1=L({u,s};p,\lambda)\Psi,\quad \Psi_2=L({v,t};q,\lambda) \Psi,
$$
where
$$
L(u,s;p,\lambda) =
\begin{pmatrix}
    0&0& -u  &s(u+p-\lambda)              \\
    0&0& -1 & s             \\
    -1 &u+p-\lambda & 0 &    0          \\
     -s& -s u& 0 &  0
\end{pmatrix},
$$
\item the companion map $\phi^c: (u,s,v_1,t_1)\mapsto (u_2,s_2,v,t)$ is a Yang-Baxter map.
\end{enumerate}
\end{prop}
\begin{proof}
\noindent The statements $(1)$ and $(2)$ can be proven by direct calculations.\\
\noindent To prove statement $(3)$   we observe  that the compatibility of the system (\ref{multi-f4}) is manifested by the formulae which is symmetric under the interchange $2$ to $3$
\begin{align*}
X^1_{23}=\frac{1}{\chi^{23}\psi^{23}}\left(X^1\psi^{23}(\chi^{23}+\pi^{23})-\pi^{23}(X^3Y^2-X^2Y^3)-\frac{\pi^{23}\chi^{12}\chi^{13}\psi^{23}(\chi^{23}+\pi^{23})}
{\pi^{23}X^1-p^1\chi^{23}+p^3X^2-p^2X^3}\right),\\
Y^1_{23}=\frac{\chi^{23}}{\chi^{23}+\pi^{23}}\left(Y^1-\frac{\pi^{23}\psi^{12}\psi^{13}}{\pi^{23}Y^1-p^1\psi^{23}+p^3Y^2-p^2Y^3}\right),
\end{align*}
where for the sake of brevity we introduced the notation $\chi^{lm}:=X^l-X^m,$ $\psi^{lm}:=Y^l-Y^m,$ $\pi^{lm}:=p^l-p^m,$ $l,m\in\{1,2,3\}.$\\
\noindent To prove $(4),$ one has to use similar arguments as in \cite{FrankABS2} where Lax representations where derived for 3D-compatible quad-equations. Specifically, $m\mathcal{H}1$ as well as all members of the $m-$list and the $a-$list (as we shall see in the next section), are of the following special form
\be \label{sys-gen}
\begin{array}{ll}
{\dis u_2=\frac{a^1(v,t)+a^2(v,t)s}{a^3(v,t)+a^4(v,t)s}},& {\dis s_2=\frac{b^1(v,t)+b^2(v,t)u}{b^3(v,t)+b^4(v,t)u},}\\ [3mm]
{\dis v_1=\frac{c^1(u,s)+c^2(u,s)t}{c^3(u,s)+c^4(u,s)t}},& {\dis t_1=\frac{d^1(u,s)+d^2(u,s)s}{d^3(u,s)+d^4(u,s)d}},
\end{array}
\ee
where $a^i,b^i, c^i,$ and $d^i,\; i=1,\ldots,4$ are functions of the indicated variables that might depend also on  $p$ and $q$. Moreover, it can be written as a projective action
\be\label{lax-rep}
 U_2= M(v,t,q;p)[ U],\quad  V_1= L(u,s,p;q)[ V],
\ee
where
$$
U=\begin{pmatrix}
u\\
1\\
s\\
1
\end{pmatrix},\quad V=\begin{pmatrix}
v\\
1\\
t\\
1
\end{pmatrix},\quad U_2=\begin{pmatrix}
u_2\\
1\\
s_2\\
1
\end{pmatrix},\quad V_1=\begin{pmatrix}
v_1\\
1\\
t_1\\
1
\end{pmatrix},
$$
and
$$
 M(v,t,q,p)=\begin{pmatrix}
                  0&0&\mathcal{B}^1 a^2&\mathcal{B}^1 a^1\\
                  0&0&\mathcal{B}^1 a^4&\mathcal{B}^1 a^3\\
                  \mathcal{B}^2 b^2&\mathcal{B}^2 b^1&0&0\\
                  \mathcal{B}^2 b^4&\mathcal{B}^2 b^3&0&0
                  \end{pmatrix}, \quad  L(u,s,p,q)=\begin{pmatrix}
                  0&0&\mathcal{A}^1 c^2&\mathcal{A}^1 c^1\\
                  0&0&\mathcal{A}^1 c^4&\mathcal{A}^1 c^3\\
                  \mathcal{A}^2 d^2&\mathcal{A}^2 d^1&0&0\\
                  \mathcal{A}^2 d^4&\mathcal{A}^2 d^3&0&0
                  \end{pmatrix},
$$
with $\mathcal{A}^1,\mathcal{A}^2,\mathcal{B}^1$  and $\mathcal{B}^2$ appropriate gauges. Due to the 3D-compatibility of the  system (\ref{sys-gen}),  $M^T(v,t,q;\lambda)$ or $L(u,s,p;\lambda)$ serve as Lax matrices  for the system of difference equations (\ref{sys-gen}) \cite{Bobenko2002,Veselov:2003b}.
\\
\noindent Fact $(5)$ is a consequence of the compatibility conditions $X^i_{jk}=X^i_{kj}, \; Y^i_{jk}=Y^i_{kj},$ for  the map $\phi^c$ is Yang-Baxter (see \cite{ABSf}).
\end{proof}
The system $m\mathcal{H}1$ defines a map $\phi: \mathbb{R}^4\ni (u,s,v,t)\mapsto (u_2,s_2,v_1,t_1)\in\mathbb{R}^4.$
\begin{prop} \label{prop2}
For the map $\phi$ it holds
\begin{enumerate}
\item it is measure preserving with density $m(u,s,v,t)=n(u,s,v,t)d(u,s,v,t)\footnote{A map $\phi:(u,s,v,t)\mapsto (u_2,s_2,v_1,t_1)$ is called {\it measure preserving map} with density $m(u,s,v,t)$, iff $m(u_2,s_2,v_1,t_1)=\frac{\partial (u_2,s_2,v_1,t_1)}{\partial (u,s,v,t)} m(u,s,v,t),$ where $\frac{\partial (u_2,s_2,v_1,t_1)}{\partial (u,s,v,t)}$ the Jacobian determinant of the mapping $\phi.$}$ where                                                \begin{align*}
    n(u,s,v,t)=(p-q+u-v)t-s(u-v),& & d(u,s,v,t)=(p-q+u-v)s-t(u-v),
    \end{align*}
\item it preserves the Poisson structure
\begin{align*}
\Omega= n(u,s,v,t)\frac{s}{t}\frac{\partial}{\partial u}\wedge \frac{\partial}{\partial s}-d(u,s,v,t)\frac{\partial}{\partial u}\wedge \frac{\partial}{\partial t}-n(u,s,v,t)\frac{s}{t}\frac{\partial}{\partial s}\wedge \frac{\partial}{\partial v}-d(u,s,v,t)\frac{\partial}{\partial v}\wedge \frac{\partial}{\partial t},
\end{align*}
\item it is a Liouville integrable map.
\end{enumerate}
\end{prop}

\begin{prop} \label{prop:3.3}
The bond system $m\mathcal{H}1,$ corresponds to the following 3D-compatible vertex system
\begin{align*}
&2(z_{12}-z)(x_2-x_1)+(p-q)(x_1+x_2)=0,&2(x_{12}-x)(z_2-z_1)+(p-q)(x_{12}+x)=0,&&
\end{align*}
which we refer to as $double-H2$ or shortly as $d-H2$.
\end{prop}
\begin{proof}
The first invariant condition $\frac{s_2}{t_1}=\frac{s}{t},$ guarantees the existence of a potential function $x$ s.t.
\begin{align*}
&s=\frac{x_1}{x},& t=\frac{x_2}{x}.
\end{align*}
The second invariant suggests $p - q + u + u_2 - v - v_1=0,$  which in turn guarantees the existence of a potential function $z$ s.t.
\begin{align*}
&u=z_1+z-\frac{p}{2},& v=z_2+z-\frac{q}{2}.
\end{align*}
Finally, the difference system $m\mathcal{H}1$ in terms of the potential functions $x,z,$  reads exactly as the indicated by the proposition system $d-H2$. The 3D-compatibility of the vertex system $d-H2,$ is manifested by the formulae:
\begin{align}
x_{123}=-\frac{x_1x_2(p_1-p_2)+x_2x_3(p_2-p_3)+x_3x_1(p_3-p_1)}{x_1(p_2-p_3)+x_2(p_3-p_1)+x_3(p_1-p_2)},\\
z_{123}=-\frac{(p_1-p_2)(p_2-p_3)(p_3-p_1)+z_1z_2(p_1-p_2)+z_2z_3(p_2-p_3)+z_3z_1(p_3-p_1)}{z_1(p_2-p_3)+z_2(p_3-p_1)+z_3(p_1-p_2)}.
\end{align}
Note that in the 3d-compatibility formulae for $x_{123}, z_{123}$ the variables $x$ and $z$ separate! The formula for $x_{123}$ coincides with the 3D-compatibility formulae of $H1,$ whereas the formula for $z_{123}$ coincides with the 3D-compatibility formulae of $H2$.
\end{proof}

\begin{lemma}
The constraint $tu=sv,$ or equivalently $s_2u_2=t_1v_1,$ is consistent with the system $m\mathcal{H}1.$
\end{lemma}

\begin{figure} 
\def\svgwidth{0.37\linewidth}
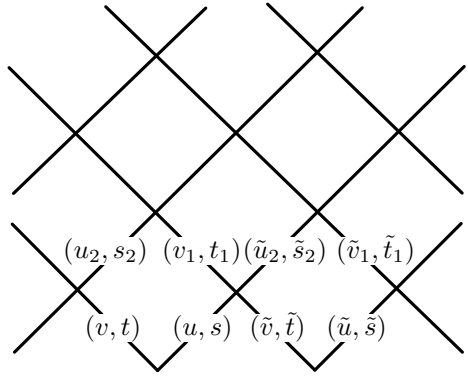
\caption{Consistent constraint on the lattice} \label{figx}
\end{figure}

\begin{proof}
It is enough to prove that if
\begin{equation}\label{eg-lemma-comp}
tu=sv\quad  \mbox{and} \quad \tilde{t}\tilde{u}=\tilde{s}\tilde{v}
 \end{equation}
 holds on two consecutive quads with a common vertex of the $\mathbb{Z}^2$  lattice, then   $t_1\tilde{u_2}=\tilde{s_2}v_1$  holds  (see Figure \ref{figx}). To prove this statement, first note that if $(\ref{eg-lemma-comp})$ holds, from the difference system   $m\mathcal{H}1$ and its tilde version, namely from:
 $$
\begin{array}{ll}
\begin{array}{l}
{\dis u_2=v+t\frac{p-q}{s-t},\quad v_1=u+s\frac{p-q}{s-t},} \\
{\dis s_2=\frac{1}{t}+\frac{p-q}{t(u-v)},\quad t_1=\frac{1}{s}+\frac{p-q}{s(u-v)}},
\end{array}& \begin{array}{l}
{\dis \tilde{u}_2=\tilde{v}+\tilde{t}\frac{p-q}{\tilde{s}-\tilde{t}},\quad \tilde{v}_1=\tilde{u}+\tilde{s}\frac{p-q}{\tilde{s}-\tilde{t}},} \\
{\dis \tilde{s}_2=\frac{1}{\tilde{t}}+\frac{p-q}{\tilde{t}(\tilde{u}-\tilde{v})},\quad \tilde{t}_1=\frac{1}{\tilde{s}}+\frac{p-q}{\tilde{s}(\tilde{u}-\tilde{v})}},
\end{array}
\end{array}
  $$
  we obtain that
\begin{equation}\label{eg-lemma-comp2}
s_2u_2=t_1v_1,\quad   \tilde{s}_2\tilde{u}_2=\tilde{t}_1\tilde{v}_1.
\end{equation}
Using similar arguments, $us=\tilde{v}\tilde{t}$ also holds.

The expression
\be\label{lemma1-pf}
\frac{t_1\tilde{u_2}}{\tilde{s_2}v_1},
\ee
 is a rational function of $u,s,v,t,\tilde{u},\tilde{s},\tilde{v}$ and $\tilde{t}.$ Then from relations $(\ref{eg-lemma-comp})$ we obtain
 $$
 \frac{t_1\tilde{u_2}}{\tilde{s_2}v_1}=\frac{us}{\tilde{v}\tilde{t}},
 $$
 that equals $1$ since $su=\tilde{v}\tilde{t}$ and that completes the proof.

\end{proof}

\begin{prop}
The system $m\mathcal{H}1$ together with the constraint $tu=sv,$ or equivalently with the constraint $s_2u_2=t_1v_1,$  leads to the $H1$ integrable lattice equation.
\end{prop}
\begin{proof}
If we impose $tu=sv,$ or equivalently  $s_2u_2=t_1v_1,$ then it can be easily shown that the system $m\mathcal{H}1$ satisfies the following additional invariant conditions:
\begin{align} \label{con2}
&u_2u=v_1v,& \frac{s_2}{s}=\frac{t_1}{t}&.
\end{align}
The relations (\ref{con2}) guarantee the existence of a potential function $x$ such that:
\begin{align} \label{con2}
&u=x_1x,&v=x_2x,&&s=\frac{x_1}{x},&&t=\frac{x_2}{x}.
\end{align}
Then $m\mathcal{H}1,$ 
written in terms of $x$ is exactly the $H1$ integrable lattice equation.
\end{proof}

\section{Compatible systems of difference equations in bond variables} \label{Section:4}
In this section we obtain from the list of integrable quad-equations mentioned in the introduction, two lists which we refer to as  the $m-$list (where $m$ stands as abbreviation of $multiplicative$) and the $a-$list (where $a$ stands as abbreviation of $additive$) respectively, of integrable difference systems in bond variables.  

\subsection{The $m-$list of difference systems in bond variables}

We consider any member $Q(x,x_1,x_2,x_{12};p,q)=0,$ of the list of integrable quad equations presented in the introduction,  apart $Q2$ and $H2$\footnote{We do not consider $Q2$ and $H2$ since for these members  our procedure leads to multi-quadratic relations rather than maps. } and the set of functions
$$
\begin{array}{l}
\{u,s,v,t\},\;\;\mbox{where}\;\;u:=x_1 x,\quad v:=x_2 x, \quad s:=x_1/x,\quad t:=x_2/x,\\
\{u_2,s_2,v_1,t_1\},\;\;\mbox{where}\;\;u_2:=x_{12} x_2,\quad v_1:=x_{12} x_1, \quad s_2:=x_{12}/x_2,\quad t_1:=x_{12}/x_1.
\end{array}
$$
The functions  are not independent, the following relations hold
\be \label{depm0}
u_2u=v_1v,\quad \frac{s_2}{s}=\frac{t_1}{t},\quad tu=sv.
\ee
Using the above functionally independent relations, the lattice equation $Q(x,x_1,x_2,x_{12};p,q)=0$ under consideration can be re-written in four different ways in terms of $u,s,v,t$ and their shifts.
 Among the various ways that we can re-write the quad-equation under consideration, we restrict our attention to the
 \be \label{tar-30}
u_2=f^1(s,v,t),\quad s_2=f^2(u,v,t),\quad v_1=f^3(u,s,t),\quad t_1=f^4(u,s,v),
\ee
where $f^i,\;i=1,2,3,4$ are appropriate rational functions depending on the indicated variables. We shall now make this requirement more precise and look for some consequences.
 The system of difference equations $(\ref{tar-30})$ together with the relations $(\ref{depm0}),$ is clearly integrable since  it is the quad-equation under consideration re-written as a difference system in bond variables.  If we consider though just the system $(\ref{tar-30})$ without the relations $(\ref{depm0}),$  since as it can be easily shown the relations $(\ref{depm0})$ are not consequences of $(\ref{tar-30}),$ integrability of any sort is not expected. Miraculously, the systems of equations $(\ref{tar-30})$ that arise from any member of the list of quad-equations apart  $Q2$ and  $H2$ integrable lattice equations, even without the relations $(\ref{depm0}),$ are integrable! The key integrability feature of these systems  is the multi-dimensional compatibility that is equivalent to  the Lax formulation of the latter.

As an example let us derive $m\mathcal{Q}4$ from $Q4.$  $Q4$ in terms of $u_2,s,v,t$ reads:
\be \label{q4u}
u_2=\frac{at(pq-sv)+v(ps-qt)}{a(1-pqsv)+qs-pt},
\ee
where $a:=\frac{pQ-Pq}{1-p^2q^2}.$ Also $Q4$ in terms of $s_2,u,v,t,$ of $v_1,u,s,t,$ and in terms of $t_1,u,v,s,$ respectively reads:
\be \label{q4svt}
\begin{array}{lll}
s_2={\dis \frac{a(pq-t u)+pu-qv}{a v (1-p q t u)+t(q u-p v)}},& v_1={\dis\frac{as(pq-t u)+u(p s-q t)}{a(1-p q t u)q s-p t}},&
t_1={\dis\frac{a(-pq+sv)-pu+qv}{au(-1+p q s v)+s(p v-q u)}}.
\end{array}
\ee
Equations $(\ref{q4u}),(\ref{q4svt}),$ without the relations $(\ref{depm0}),$ constitutes the $m\mathcal{Q}4$ difference system in bond variables, which
after the identification $u=X^1, V=X^2, u_2=X^1_2, v_1=X^2_1, \; p=p^1,\; q=p^2,$ gets exactly the form  in Proposition \ref{prop3.1}.
Working similarly with the remaining members of the list, we obtain difference systems  which constitute the {\it $m-$list}  of integrable difference systems in bond variables.

\begin{prop}[The $m$-list of difference systems in bond variables]\label{prop3.1}
The following systems of difference equations:
\begin{align*}
&\begin{aligned}
&X^i_j=\frac{aY^j(p^ip^j-Y^iX^j)+X^j(p^iY^i-p^jY^j)}{a(1-p^ip^jX^jY^j)+p^jY^i-p^iY^j},& &Y^i_j=\frac{a(p^ip^j-X^iY^j)+p^iX^i-p^jX^j}{aX^j(1-p^ip^jX^iY^j)+Y^j(p^jX^i-P^iX^j)}&,\\
&\mbox{where}\;\; a:=\frac{p^iP^j-p^jP^i}{1-[p^i]^2[p^j]^2},&& [P^i]^2=[p^i]^4-\gamma [p^i]^2+1,\; [P^j]^2=[p^j]^4-\gamma [p^j]^2+1,&
\end{aligned}& (m\mathcal{Q}4)\\
&\begin{aligned}
&X^i_j=\frac{\delta abc Y^j+X^jY^j(b+cY^i)-aX^jY^i}{aY^j-bY^i-c},& Y^i_j=\frac{\delta abc-aX^i+bX^j+cX^iY^j}{aX^jY^j-bX^iY^j-cX^j},& \\
&\mbox{where}\;\;a:=p^i-\frac{1}{p^i},\;\; b:=p^j-\frac{1}{p^j},\;\; c:=\frac{p^i}{p^j}-\frac{p^j}{p^i},&& {} &
\end{aligned} & (m\mathcal{Q}3^{\delta})\\
&\begin{aligned}
&X^i_j=X^j-(p^i-p^j)\frac{-\delta p^ip^jY^j+X^j(1-Y^i)(1-Y^j)}{p^i(1-Y^j)-p^j(1-Y^i)},& \\
&Y^i_j=\frac{1}{Y^j}\left(1+(p^i-p^j)\frac{-\delta p^ip^jY^j+(Y^j-1)(X^iY^j-X^j)}{p^iX^j(Y^j-1)+p^j(X^j-X^iY^j)}\right),&
\end{aligned} & (m\mathcal{Q}1^{\delta})\\
&\begin{aligned}
&X^i_j=\frac{p^i({p^j}^2-1)Y^iX^j+p^jY^j(X^j-p^j)+{p^i}^2Y^j(1-p^jX^j)}{p^j({p^i}^2-1)Y^i+{p^j}^2(Y^iX^j-p^iY^j)+p^i(Y^j-p^iX^jY^i},& \\
&Y^i_j=\frac{1}{Y^j}\frac{p^i({p^j}^2-1)X^i+p^j(X^j-p^j)+{p^i}^2(1-p^jX^j)}{p^j({p^i}^2-1)X^i+{p^j}^2(X^i-p^i)X^j+p^i(1-p^iX^i)X^j},&
\end{aligned} & (m\mathcal{A}2)\\
&\begin{aligned}
&X^i_j=X^j+(p^i-p^j)\frac{\delta p^ip^jY^j-X^j(1+Y^i)(1+Y^j)}{p^i(1+Y^j)-p^j(1+Y^i)},& \\
&Y^i_j=\frac{1}{Y^j}\left(1+(p^i-p^j)\frac{\delta p^ip^jY^j-(1+Y^j)(X^iY^j+X^j)}{p^iX^j(Y^j+1)+p^j(X^j+X^iY^j)}\right),&
\end{aligned} & (m\mathcal{A}1^{\delta})\\
&\begin{aligned}
&X^i_j=\frac{\delta^2({p^j}^2-{p^i}^2)Y^j+X^j(p^jY^j-p^iY^i)}{p^iY^j-p^jY^i},& \\
&Y^i_j=\frac{1}{Y^j}\frac{\delta^2({p^i}^2-{p^j}^2)+p^iX^i-p^jX^j}{p^jX^i-p^iX^j},&
\end{aligned} & (m\mathcal{H}3^{\delta})\\
&\begin{aligned}
&X^i_j=X^j+Y^j\frac{p^i-p^j}{Y^i-Y^j},&\\
&Y^i_j=\frac{1}{Y^j}\left(1+\frac{p^i-p^j}{X^i-X^j}\right),&
\end{aligned} & (m\mathcal{H}1)
\end{align*}
where $i\neq j\in\{1,2\},$
\begin{enumerate}
\item admit the invariance conditions presented in Table \ref{table13}
\item it holds $X^iY^j-X^jY^i=0,\; i\neq j\in\{1,2\}$ if and only if $X^i_jY^i_j-X^j_iY^j_i=0,\; i\neq j\in\{1,2\}$
\item they can be extended to multi-dimensions by allowing $i\neq j\neq k\in\{1,2,\ldots,n\},$
 they are multi-dimensional compatible i.e. $X^i_{jk}=X^i_{kj}, \; Y^i_{jk}=Y^i_{kj}$ $i\neq j\neq k\in\{1,2,\ldots,n\}$ 
\end{enumerate}
\end{prop}

\begin{remark}
The point transformation $(X^i,Y^i)\mapsto (-X^i,-Y^i)$ maps $m\mathcal{Q}1^{\delta}$ to $m\mathcal{A}1^{\delta},$ so these bond systems are point equivalent. $m\mathcal{A}2$ is also point equivalent to $m\mathcal{Q}3^{0},$ the transformation in this case reads $(X^i,Y^i)\mapsto \left([Y^i]^{(-1)^{1+m_1+m_2+\ldots +m_n}},[X^i]^{(-1)^{1+m_1+m_2+\ldots +m_n}}\right).$ Although $m\mathcal{Q}1^{\delta}$ is essentially the same as  $m\mathcal{A}1^{\delta}$  and $m\mathcal{A}2$ is essentially the same as $m\mathcal{Q}3^{0},$ for historical reasons we decided to keep all these pairwise equivalent bond systems as members of the $m-$list.
\end{remark}

\begin{remark}
For each one of the difference systems of the $m-$list one can associate a rational map. Namely,  for $n=2,$ by making the identifications $u:=X^1, v:=X^2, s:=Y^1, t:=Y^2$ and $u_2:=X^1_2, v_1:=X^2_1, s_2:=Y^1_2, t_1:=Y^2_1,$ to any of the difference systems presented earlier, we have an associated map $\phi: (u,s,v,t)\mapsto (u_2,s_2,v_1,t_1).$  Due to the compatibility of the difference systems, the companion  map $\phi_c: (u,s,v_1,t_1)\mapsto (u_2,s_2,v,t)$ of the map $\phi$, is a Yang-Baxter map (see \cite{ABSf}).
 In Appendix \ref{appA}, we present explicitly the corresponding Yang-Baxter maps associated with the difference systems of the $m-$list.

\end{remark}

\subsection{The $a-$list of difference systems in bond variables}

We consider the set of functions
$$
\begin{array}{l}
\{u,s,v,t\},\;\;\mbox{where}\;\;u:=x_1+x,\quad v:=x_2+x, \quad s:=x_1-x,\quad t:=x_2-x,\\
\{u_2,s_2,v_1,t_1\},\;\;\mbox{where}\;\;u_2:=x_{12}+x_2,\quad v_1:=x_{12}+x_1, \quad s_2:=x_{12}-x_2,\quad t_1:=x_{12}-x_1.
\end{array}
$$
These functions  are not independent,  the following relations hold
\be \label{depa}
u_2+u=v_1+v,\quad s_2-s=t_1-t,\quad t+u=s+v.
\ee
Following the procedure described in the previous subsection, by using these  sets of functions  for any of the quad-equations presented in the introduction, apart $Q4$ and  $A2\footnote{Applying the above procedure to $Q4$ and $A2$ we obtain difference systems which are not multidimensional compatible unless we take in account the relations $(\ref{depa})$ as well. So, although we can consider the derived systems as members of the $a-$list, we chose not to since they can be incorporated into the corresponding quad-graph equation listed in the Introduction.},$ we obtain
\be \label{tar-4}
u_2=g^1(s,v,t),\quad s_2=g^2(u,v,t),\quad v_1=g^3(u,s,t),\quad t_1=g^4(u,s,v),
\ee
where $g^i,\;i=1,2,3,4$ are appropriate rational functions depending on the indicated variables. The system of difference equations $(\ref{tar-4})$ together with the relations $(\ref{depa}),$ is clearly integrable since  it is the underlying quad-equation under consideration re-written as a difference system in bond variables.
 Again, if we consider just the systems $(\ref{tar-4})$ without the relations $(\ref{depa}),$ we have a more general system than the underlying quad-equation and this  system is integrable. We refer to the list of these superior systems as the $a-$list of integrable difference systems in bond variables.

As an example let us derive $m\mathcal{Q}2$ from $Q2.$  $Q2$ in terms of $u_2,s,v,t$ reads:
\be \label{q2u}
u_2=v+(p-q)\frac{pq(p^2-pq+q^2)+st-pq(s-t+2v)}{pq(p-q)+pt-qs}
\ee
 Also $Q2$ in terms of $s_2,u,v,t,$ of $v_1,u,s,t,$ and in terms of $t_1,u,v,s,$ respectively reads:
\be \label{q2svt}
\begin{array}{l}
s_2={\dis p\frac{q(p-q)(p^2-pq+q^2-t-u-v)+t(u-v)}{(p-q)(pq+t)-q(u-v)}},\\
v_1={\dis u+(p-q)\frac{pq(p^2-pq+q^2)+st-pq(t-s+2u)}{pq(p-q)+pt-qs}},\\
t_1={\dis q\frac{p(p-q)(p^2-pq+q^2-s-u-v)+s(u-v)}{(p-q)(pq+s)-p(u-v)}}.
\end{array}
\ee
Equations $(\ref{q2u}),(\ref{q2svt}),$ without the relations $(\ref{depa}),$ constitutes the $a\mathcal{Q}2$ difference system in bond variables.
Working similarly with the remaining members of the list of quad-equations, we obtain difference systems  which constitute the {\it $a-$list}  of integrable difference systems in bond variables.

\begin{prop}\label{prop3.2}
The following systems of difference equations:
\begin{align*}
&\begin{aligned}
&X^i_j= Y^i-2\frac{\delta a b c+(bY^i-aX^j)(Y^i-Y^j)+cY^iX^j}{2bY^i-(a+b+c)Y^j-(a-b-c)X^j},&& \\
&Y^i_j=X^i-2\frac{\delta a b c+(bX^i-aX^j)(X^i+Y^j)+cX^iX^j}{2bX^i-(a-b+c)Y^j-(a+b-c)X^j},&&\\
&\mbox{where}\;\;a:=p^i-\frac{1}{p^i}, b:=p^j-\frac{1}{p^j}, c:=\frac{p^i}{p^j}-\frac{p^j}{p^i}, &&
\end{aligned} & (a\mathcal{Q}3^{\delta})\\
&\begin{aligned}
&X^i_j=X^j+(p^i-p^j)\frac{p^ip^j([p^i]^2-p^ip^j+[p^j]^2)+Y^iY^j-p^ip^j(Y^i-Y^j+2X^j)}{p^iY^j-p^jY^i+p^ip^j(p^i-p^j)},&&\\
&Y^i_j=p^i\frac{p^j(p^i-p^j)([p^i]^2-p^ip^j+[p^j]^2-X^i-X^j-Y^j)+Y^j(X^i-X^j)}{(p^ip^j+Y^j)(p^i-p^j)-p^j(X^i-X^j)},&&
\end{aligned} & (a\mathcal{Q}2)\\
&\begin{aligned}
&X^i_j=X^j+(p^i-p^j)\frac{\delta p^ip^j-Y^iY^j}{p^jY^i-p^iY^j},&& \\
&Y^i_j=p^i\frac{\delta p^j(p^j-p^i)+Y^j(X^i-X^j)}{p^iY^j-p^j(X^i-X^j+Y^j)},&&
\end{aligned} & (a\mathcal{Q}1^{\delta})\\
&\begin{aligned}
&X^i_j=p^i\frac{\delta p^j(p^i-p^j)+X^j(Y^j-Y^i)}{p^iX^j-p^j(Y^i+X^j-Y^j)},&& \\
&Y^i_j=-Y^j+(p^i-p^j)\frac{\delta p^ip^j-X^iX^j}{p^iX^j-p^jX^i},&&
\end{aligned} & (a\mathcal{A}1^{\delta})\\
&\begin{aligned}
&X^i_j=\frac{(p^i-p^j)\left(2\delta (p^i+p^j)-Y^iY^j\right)+\left((p^i+p^j)Y^i-2p^iY^j\right)X^j}{p^j(2Y^i+X^j-Y^j)-p^i(X^j+Y^j)},&& \\
&Y^i_j=\frac{2\delta ({p^j}^2-{p^i}^2)+p^iX^iY^j-p^iX^j(X^i+2Y^j)+p^jX^i(X^j+Y^j)}{p^i(X^j+Y^j)-p^j(2X^i-X^j+Y^j)},&&
\end{aligned} & (a\mathcal{H}3^{\delta})\\
&\begin{aligned}
&X^i_j=X^j-(p^i-p^j)\frac{p^i+p^j+Y^i-Y^j+2X^j}{p^i-p^j+Y^i-Y^j},&& \\
&Y^i_j=p^i-p^j-Y^j-2(p^i-p^j)\frac{p^i+X^i}{p^i-p^j+X^i-X^j},&&
\end{aligned} & (a\mathcal{H}2)\\
&\begin{aligned}
&X^i_j=X^j+\frac{p^i-p^j}{Y^i-Y^j},&&\\
&Y^i_j=-Y^i+\frac{p^i-p^j}{X^i-X^j},&&
\end{aligned} & (a\mathcal{H}1)
\end{align*}
where $i\neq j\in\{1,2\},$
\begin{enumerate}
\item admit the invariance conditions presented in Table \ref{table14}
\item it holds $X^i+Y^j-X^j-Y^i=0,\; i\neq j\in\{1,2\}$ if and only if $X^i_j+Y^i_j-X^j_i-Y^j_i=0,\; i\neq j\in\{1,2\}$
\item they can be extended to multi-dimensions by allowing $i\neq j\neq k\in\{1,2,\ldots,n\},$
 they are multi-dimensional compatible i.e. $X^i_{jk}=X^i_{kj}, \; Y^i_{jk}=Y^i_{kj}$ $i\neq j\neq k\in\{1,2,\ldots,n\}$ 
\end{enumerate}

\end{prop}

\begin{remark}
The point transformation $(X^i,Y^i)\mapsto \left((-1)^{1+m_1+m_2+\ldots+m_n}Y^i,(-1)^{1+m_1+m_2+\ldots+m_n}X^i\right)$ maps $a\mathcal{Q}1^{\delta}$ to $a\mathcal{A}1^{\delta},$ so these bond systems are point equivalent. Although $a\mathcal{Q}1^{\delta}$ is essentially the same as  $a\mathcal{A}1^{\delta}$, for historical reasons we decided to keep both bond systems as members of the $a-$list.
\end{remark}

\begin{remark}
For each one of the difference systems of the $a-$list one can associate a rational map. Namely,  for $n=2,$ by making the identifications $u:=X^1, v:=X^2, s:=Y^1, t:=Y^2$ and $u_2:=X^1_2, v_1:=X^2_1, s_2:=Y^1_2, t_1:=Y^2_1,$ to any of the difference systems presented earlier, we have an associated map $\phi: (u,s,v,t)\mapsto (u_2,s_2,v_1,t_1).$  Due to the compatibility of the difference systems, the companion  map $\phi_c: (u,s,v_1,t_1)\mapsto (u_2,s_2,v,t)$ of the map $\phi$, is a Yang-Baxter map (see \cite{ABSf}).
 In Appendix \ref{appB}, we present explicitly the corresponding Yang-Baxter maps associated with the difference systems of the $a-$list.

\end{remark}

\subsection{Reductions to known integrable quad-equations}

\begin{lemma}
The set of constraints $X^iY^j-X^jY^i=0,\; i\neq j=1,\ldots, n$ or the set $X^i_jY^i_j-X^j_iY^j_i=0,\; i\neq j=1,\ldots, n$ are consistent with any of the difference systems of equations of Proposition \ref{prop3.1}.
\end{lemma}

\begin{lemma}
The set of constraints $X^i+Y^j-X^j-Y^i=0,\; i\neq j=1,\ldots, n$ or the set $X^i_j+Y^i_j-X^j_i-Y^j_i=0,\; i\neq j=1,\ldots, n$ are consistent with any of the difference systems of equations of Proposition \ref{prop3.2}.
\end{lemma}

For the rest of this subsection we consider $n=2$ and we make the identifications $u:=X^1, v:=X^2, s:=Y^1, t:=Y^2$ and $u_2:=X^1_2, v_1:=X^2_1, s_2:=Y^1_2, t_1:=Y^2_1.$

\begin{prop}
The systems of difference equations $m\mathcal{Q}4$, $m\mathcal{Q}3^{\delta}$, $m\mathcal{Q}1^{\delta}$, $m\mathcal{A}2$, $m\mathcal{A}1^{\delta}$, $m\mathcal{H}3^{\delta}$,and $ m\mathcal{H}1$ for $n=2$ of Proposition \ref{prop3.1} together with the constraint $tu=sv,$ or  with the constraint $s_2u_2=t_1v_1,$  reduce to the $Q4$, $Q3^{\delta}$, $Q1^{\delta}$, $A2$, $A1^{\delta}$, $H3^{\delta}$, and  $H1 $ integrable lattice equations respectively.
\end{prop}
\begin{proof}
If we impose $tu=sv,$ or   $s_2u_2=t_1v_1,$ then it can be easily shown that any of the $m\mathcal{Q}4$, $m\mathcal{Q}3^{\delta}$, $m\mathcal{Q}1^{\delta}$, $m\mathcal{A}2$, $m\mathcal{A}1^{\delta}$, $m\mathcal{H}3^{\delta}$,and $ m\mathcal{H}1$ satisfies the following additional invariant conditions:
\begin{align} \label{con22}
&u_2u=v_1v,& \frac{s_2}{s}=\frac{t_1}{t}&.
\end{align}
The relations (\ref{con22}) guarantee the existence of a potential function $x$ such that:
\begin{align} \label{con23}
&u=x_1x,&v=x_2x,&&s=\frac{x_1}{x},&&t=\frac{x_2}{x}.
\end{align}
Then $m\mathcal{Q}4$, $m\mathcal{Q}3^{\delta}$, $m\mathcal{Q}1^{\delta}$, $m\mathcal{A}2$, $m\mathcal{A}1^{\delta}$, $m\mathcal{H}3^{\delta}$,and $ m\mathcal{H}1$ written in terms of $x,$ coincide with the $Q4$, $Q3^{\delta}$, $Q1^{\delta}$, $A2$, $A1^{\delta}$, $H3^{\delta}$ and $H1$ quad-equations.
\end{proof}

\begin{prop}
The systems of difference equations $ a\mathcal{Q}3^{\delta}$, $a\mathcal{Q}2$, $a\mathcal{Q}1^{\delta}$,  $a\mathcal{A}1^{\delta}$, $a\mathcal{H}3^{\delta}$, $a\mathcal{H}2$ and $a\mathcal{H}1$ for $n=2$ of Proposition \ref{prop3.2} together with the constraint $t+u=s+v,$ or  with the constraint $s_2+u_2=t_1+v_1,$  reduce to the $ Q3^{\delta}$, $Q2$, $Q1^{\delta}$,  $A1^{\delta}$, $H3^{\delta}$, $H2$ and $H1 $ integrable lattice equations respectively.
\end{prop}
\begin{proof}
If we impose $t+u=s+v,$ or   $s_2+u_2=t_1+v_1,$ then it can be easily shown that any of the $ a\mathcal{Q}3^{\delta}$, $a\mathcal{Q}2$, $a\mathcal{Q}1^{\delta}$,  $a\mathcal{A}1^{\delta}$, $a\mathcal{H}3^{\delta}$, $a\mathcal{H}2$ and $a\mathcal{H}1$ satisfies the following additional invariant conditions:
\begin{align} \label{con22}
&u_2-v_1=v-u,& s_2-t_1=s-t&.
\end{align}
The relations (\ref{con22}) guarantee the existence of a potential function $x$ such that:
\begin{align} \label{con23}
&u=x_1+x,&v=x_2+x,&&s=x_1-x,&&t=x_2-x.
\end{align}
Then $ a\mathcal{Q}3^{\delta}$, $a\mathcal{Q}2$, $a\mathcal{Q}1^{\delta}$,  $a\mathcal{A}1^{\delta}$, $a\mathcal{H}3^{\delta}$, $a\mathcal{H}2$ and $a\mathcal{H}1$ written in terms of $x,$ coincide with the $ Q3^{\delta}$, $Q2$, $Q1^{\delta}$,  $A1^{\delta}$, $H3^{\delta}$, $H2$ and $H1$ quad-equations.
\end{proof}

\section{Lax pairs of the $m-$list and the $a-$list of difference systems in Bond Variables} \label{Section:5}

We recall the following definition:
\begin{dfn} \label{dfn111}
 A matrix  $L(u,s;p,\lambda)$ is called a Lax matrix of  any of the difference systems of Propositions \ref{prop3.2} and \ref{prop3.1} with $n=2$, if the equations
\be \label{dif-g-s}
u_2=f^1(s,v,t),\quad s_2=f^2(u,v,t),\quad v_1=f^3(u,s,t),\quad t_1=f^4(u,s,v),
\ee
where $f^i,\;\;i=1,\ldots , 4$ are rational functions of the indicated variables,
 imply that
\begin{equation} \label{lax-eq}
L({u_2,s_2};p,\lambda)\, L({v,t};q,\lambda) =
L(v_1,t_1;q,\lambda)\, L({u,s};p,\lambda)\,
\end{equation}
holds for all $\lambda$.
$L(u,s;p,\lambda)$ is called a strong Lax matrix of
$(\ref{dif-g-s})$, if the converse also holds. \\
\end{dfn}
Relation (\ref{lax-eq}) constitutes the Lax equation of the difference system in bond variables (\ref{dif-g-s}) and it stands as the compatibility condition of the linear system
$$
\Psi_1=L({u,s};p,\lambda)\Psi,\quad \Psi_2=L({v,t};q,\lambda) \Psi.
$$
\begin{prop} \label{prop_lax1}
The Lax matrices of  the difference systems of Proposition \ref{prop3.1}  reads
$$
L(u,s;p,\lambda) =
\frac{1}{\sqrt{A}}\begin{bmatrix}
    {\bf 0}     & {\bf L}^1              \\
    {\bf L}^2     & {\bf 0}
\end{bmatrix},
$$
where $\det {\bf L}^1=\det {\bf L}^2 =c A,$ with $c$ some non-zero constant. We list ${\bf L}^1, {\bf L}^2$ and $A$  in the  Table \ref{Table:1}\\

\begin{longtable}{|l|l|l|}
 \caption{The  matrices  ${\bf L}^1, \; {\bf L}^2$ and the function $A$ associated with  the difference systems of Proposition \ref{prop3.1}} \label{Table:1} \\ \toprule
 $m\mathcal{Q}4$ & $\begin{array}{l} \vspace{0.3cm}
      {\bf L}^1=\begin{bmatrix}
   \left(\lambda +\frac{p \Lambda -P \lambda}{1-p^2 \lambda ^2} s \right) u & -\left(u+\frac{p \Lambda -P \lambda}{1-p^2 \lambda ^2} \lambda \right) p s  \\
   \left(\frac{p \Lambda -P \lambda}{1-p^2 \lambda ^2} \lambda u +1 \right) p & -s \lambda -\frac{p \Lambda -P \lambda }{1-p^2 \lambda ^2} \\
\end{bmatrix}  \vspace{0.3cm} \\
       {\bf L}^2= \begin{bmatrix}
   -\lambda -\frac{p \Lambda -P \lambda}{1-p^2 \lambda ^2}  s                 &  \left(u+\frac{p \Lambda -P \lambda}{1-p^2 \lambda ^2} \lambda \right) p  \\
 -\left(\frac{p \Lambda -P \lambda}{1-p^2 \lambda ^2} \lambda u +1\right) p s &  \left(s \lambda +\frac{p \Lambda -P \lambda }{1-p^2 \lambda^2}\right) u
\end{bmatrix}\vspace{0.3cm} \\
       \mbox{where} \;\; \Lambda^2:=\lambda^4-\gamma \lambda^2+1,\;\;P^2:=p^4-\gamma p^2+1
       \end{array}$ & $A=p^2 s - u      - s u (s - p^2 u)       + 2 P s u$   \\       \midrule

     $m\mathcal{Q}3^{\delta}$   & $\begin{array}{l} \vspace{0.3cm}
      {\bf L}^1= \begin{bmatrix}
    u\left(\frac{1}{\lambda}-\lambda+s\left(\frac{\lambda}{p}-\frac{p}{\lambda}\right)\right)& s(p-\frac{1}{p})\left(u+\delta(\frac{\lambda}{p}-\frac{p}{\lambda})(\lambda-\frac{1}{\lambda})\right)   \\
    \frac{1}{p}-p     & \frac{p}{\lambda}-\frac{\lambda}{p}+s (\lambda-\frac{1}{\lambda})\end{bmatrix} \vspace{0.3cm} \\
       {\bf L}^2= \begin{bmatrix}
    \lambda-\frac{1}{\lambda}+s\left(\frac{p}{\lambda}-\frac{\lambda}{p}\right)&
    (p-\frac{1}{p})\left(-u+\delta(\frac{p}{\lambda}-\frac{\lambda}{p})(\lambda-\frac{1}{\lambda})\right)   \\
    s\left(p-\frac{1}{p}\right)     & u\left(\frac{\lambda}{p}-\frac{p}{\lambda}-s (\lambda-\frac{1}{\lambda})\right)
    \end{bmatrix}
       \end{array}$ & $A=\delta s (1-p^2)^2+pu(p-s)(1-ps)$\\  \midrule

 $m\mathcal{Q}1^{\delta}$ & $\begin{array}{l} \vspace{0.3cm}
      {\bf L}^1= \begin{bmatrix}
                   u\left(\lambda(1-s)+ps\right)& p s \left(\delta \lambda(\lambda-p)-u)\right)\\
                   p& \lambda(1-s)-p
                     \end{bmatrix} \vspace{0.3cm} \\
       {\bf L}^2= \begin{bmatrix}
                   \lambda(s-1)-ps & p  \left(\delta \lambda(p-\lambda)+u\right)\\
                   -ps& u\left(\lambda(s-1)+p\right)
                     \end{bmatrix}
       \end{array}$ & $A= \delta p^2s-u(1-s)^2$\\        \midrule

  $m\mathcal{A}2$ & $\begin{array}{l} \vspace{0.3cm}
      {\bf L}^1=  \begin{bmatrix}
                  \left(\frac{1}{p}-p\right)u& \left(\frac{p}{\lambda}-\frac{\lambda}{p}\right)s+\left(\lambda-\frac{1}{\lambda}\right)su\\
                  \frac{1}{\lambda}-\lambda+\left(\frac{\lambda}{p}-\frac{p}{\lambda}\right)u&\left(p-\frac{1}{p}\right)s
                   \end{bmatrix}  \vspace{0.3cm} \\
       {\bf L}^2=\begin{bmatrix}
                  \frac{1}{p}-p& \frac{p}{\lambda}-\frac{\lambda}{p}+\left(\lambda-\frac{1}{\lambda}\right)u\\
                  \left(\frac{1}{\lambda}-\lambda+\left(\frac{\lambda}{p}-\frac{p}{\lambda}\right)u\right)s&\left(p-\frac{1}{p}\right)su
                   \end{bmatrix}
       \end{array}$ & $A= s(p-u)(pu-1)$\\       \midrule

  $m\mathcal{A}1^{\delta}$ & $\begin{array}{l} \vspace{0.3cm}
      {\bf L}^1=\begin{bmatrix}
                u\left(\lambda(1+s)-ps\right)&ps \left(\delta\lambda(p-\lambda)-u\right)\\
                p&p-\lambda(1+s)
                \end{bmatrix} \vspace{0.3cm}\\
       {\bf L}^2=\begin{bmatrix}
                  \lambda(1+s)-ps&p(\delta  \lambda(p-\lambda)-u)\\
                  ps&u\left(p-\lambda(1+s)\right)
                  \end{bmatrix}
       \end{array}$ & $A=\delta p^2s-u(1+s)^2$ \\       \midrule

 $ m\mathcal{H}3^{\delta}$ & $\begin{array}{l} \vspace{0.3cm}
      {\bf L}^1= \begin{bmatrix}
                   \lambda u& s\left(\delta (\lambda^2-p^2)-pu\right)\\
                   p&-\lambda s
                  \end{bmatrix} , \;\;
       {\bf L}^2=\begin{bmatrix}
                  -\lambda&\delta (p^2-\lambda^2)+pu\\
                  -ps&\lambda su
                 \end{bmatrix}
       \end{array}$ & $A=s(\delta p+u) $ \\       \midrule

 $ m\mathcal{H}1$ & $\begin{array}{l} \vspace{0.3cm}
      {\bf L}^1=\begin{bmatrix}
                 -u&s\left(u+p-\lambda\right)\\
                 -1&s
      \end{bmatrix},\;\;{\bf L}^2=\begin{bmatrix}
                 -1&u+p-\lambda\\
                 -s&-su
                  \end{bmatrix} \vspace{0.3cm}\\
       \end{array}$ & $A=1$ \\    \bottomrule

\end{longtable}

\end{prop}

\begin{prop} \label{prop_lax1}
The Lax matrices of  the difference systems of Propositions  \ref{prop3.2} reads
$$
L(u,s;p,\lambda) =
\frac{1}{\sqrt{A}}\begin{bmatrix}
    {\bf 0}     & {\bf L}^1              \\
    {\bf L}^2     & {\bf 0}
\end{bmatrix},
$$
where $\det {\bf L}^1=\det {\bf L}^2 =c A,$ with $c$ some non-zero constant. We list ${\bf L}^1, {\bf L}^2$ and $A$  in the  Table \ref{Table:2}\\

\begin{longtable}{|l|l|l|}
 \caption{The  matrices  ${\bf L}^1, \; {\bf L}^2$ and the function $A$ associated with  the difference systems of Proposition \ref{prop3.2}} \label{Table:2}
 \\ \toprule
 $ a\mathcal{Q}3^{\delta}$& $\begin{array}{l} \vspace{0.3cm}
      {\bf L}^1= \begin{bmatrix}
                 \Lambda(s-u)+\left(-P+\frac{\lambda}{p}-\frac{p}{\lambda}\right)(s+u)&
                 2\Lambda\left(\delta P \left(\frac{\lambda}{p}-\frac{p}{\lambda}\right) +su\right)\\
                 -2P&\left(P+\frac{\lambda}{p}-\frac{p}{\lambda}\right)(s-u)+\Lambda(s+u)
                \end{bmatrix}\vspace{0.3cm}\\
                 {\bf L}^2=\begin{bmatrix}
               -\Lambda (s-u)-\left(P+\frac{\lambda}{p}-\frac{p}{\lambda}\right)(s+u)&2\Lambda\left(\delta P\left(\frac{p}{\lambda}-\frac{\lambda}{p}\right)+su\right)\\
               2P&\left(P+\frac{p}{\lambda}-\frac{\lambda}{p}\right)(s-u)-\Lambda(s+u)
                    \end{bmatrix} \vspace{0.3cm} \\
       \mbox{where} \;\; \Lambda:=\lambda-\frac{1}{\lambda},\;\;P:=p-\frac{1}{p}
       \end{array}$ & $\begin{array}{l}
       A=4\delta (1-p^2)^2\\
       \;\;+p(1+p)^2s^2\\
       \;\;-p(1-p)^2u^2 \end{array}$\\       \midrule

   $ a\mathcal{Q}2$      & $\begin{array}{l} \vspace{0.3cm}
      {\bf L}^1= \begin{bmatrix}
                 (\lambda  -p)(\lambda p-s)+pu&\lambda(-\lambda p+p^2+s)\left(p(p-\lambda)-u\right)+\lambda^3p(p-\lambda)\\
                 p&\lambda\left(p(p-\lambda)-s\right)
                 \end{bmatrix} \vspace{0.3cm}\\
       {\bf L}^2= \begin{bmatrix}
                  \lambda(-\lambda p+p^2+s)&\lambda(-\lambda p+p^2-s)\left(p(\lambda-p)+u\right)+\lambda^3p(\lambda-p)\\
                  -p&(\lambda-p)(\lambda p+s)+pu
                  \end{bmatrix}
       \end{array}$ &$ A=p^4+s^2-2p^2u$\\   \midrule

$ a\mathcal{Q}1^{\delta}$ & $\begin{array}{l} \vspace{0.3cm}
      {\bf L}^1=\begin{bmatrix}
                -\lambda s+p(s+u)&\lambda\left(\delta p(\lambda-p)-su\right)\\
                p&-\lambda s
      \end{bmatrix} \vspace{0.3cm}\\
       {\bf L}^2= \begin{bmatrix}
                  \lambda s& \lambda\left(\delta p(p-\lambda)-su\right)\\
                  -p& \lambda s-p(s-u)
               \end{bmatrix}
       \end{array}$ & $A=\delta p^2-s^2 $\\       \midrule

$ a\mathcal{A}1^{\delta}$ & $\begin{array}{l} \vspace{0.3cm}
      {\bf L}^1= \begin{bmatrix}
                 -\lambda u& \lambda\left(\delta p(\lambda-p)+su \right)\\
                 -p&p(s-u)+\lambda u
          \end{bmatrix} \vspace{0.3cm}\\
       {\bf L}^2= \begin{bmatrix}
                -\lambda u+p (s+u)& \lambda\left(\delta p (\lambda-p)-su\right)\\
                -p& \lambda u
               \end{bmatrix}
       \end{array}$ & $A=\delta p^2-u^2 $\\       \midrule

$ a\mathcal{H}3^{\delta}$ & $\begin{array}{l} \vspace{0.3cm}
      {\bf L}^1= \begin{bmatrix}
               \lambda (s-u)-p(s+u)& 2\delta (p^2-\lambda^2)+2\lambda s u\\
               -2p& p(s-u)+\lambda (s+u)
                              \end{bmatrix} \vspace{0.3cm} \\
       {\bf L}^2= \begin{bmatrix}
             \lambda(s-u)+p(s+u)& 2\delta (p^2-\lambda^2)-2\lambda s u \\
            -2p & -p(s-u)+\lambda (s+u)
          \end{bmatrix}
       \end{array}$ & $A=4\delta p+u^2-s^2 $\\ \midrule

$ a\mathcal{H}2$ & $\begin{array}{l} \vspace{0.3cm}
      {\bf L}^1=\begin{bmatrix}
               \lambda-p-u&\lambda^2+(s-p)(p+u)-\lambda(s-u)\\
               -1&-\lambda+p+s
              \end{bmatrix} \vspace{0.3cm}\\
       {\bf L}^2=\begin{bmatrix}
                 \lambda-p+s&\lambda^2-(p+s)(p+u)+\lambda(s+u)\\
                 -1&-\lambda+p+u
                  \end{bmatrix}
       \end{array}$ & $A=p+u$ \\ \midrule

$ a\mathcal{H}1$ & $\begin{array}{l}\vspace{0.3cm}
      {\bf L}^1=\begin{bmatrix}
              -u&-\lambda+p+su\\
              -1&s
              \end{bmatrix},\;\;
       {\bf L}^2= \begin{bmatrix}
                    s&-\lambda+p-su\\
                    -1&u
                     \end{bmatrix}
       \end{array}$ & $A=1$ \\       \bottomrule

\end{longtable}

\end{prop}

\section{Vertex systems of integrable equations} \label{Section:6}
In Proposition \ref{prop:3.3}, it was proven that by using appropriate invariants,  the bond system $m\mathcal{H}1$ leads to the vertex system of integrable equations $d-H2.$ 

Working similarly as in Proposition \ref{prop:3.3}, the bond systems $m\mathcal{A}2, m\mathcal{H}3^{\delta}, m\mathcal{H}1,$ $a\mathcal{Q}1^{\delta},$ $a\mathcal{A}1^{\delta},$ $a\mathcal{H}2$ and $a\mathcal{H}1,$ lead to two-component vertex systems of integrable equations defined on a $\mathbb{Z}^2$ lattice. This is presented in Theorem \ref{theorem1}. As for the remaining members of the $m$ and the $a-$list, the corresponding vertex systems are at the moment not known.
In the next subsection we present the invariants for all the members of the $m-$list and the $a-$list of integrable difference systems in bond variables.
\subsection{Invariants}
In the following we give the associated invariants for all members of the m-list and the a-list. They can be constructed in an algorithmic manner from the corresponding Lax matrix, by considering the coefficients of powers of the spectral parameter $\lambda$ that appear on the expression $Tr(L(u,s;p\lambda)L^{-1}(v,t;q,\lambda))$, where $Tr(A)$ stands for the trace of a matrix $A$. The invariants of the $m-$list and the $a-$list are presented in Table \ref{table13} and Table \ref{table14} respectively.
\begin{table}[!h]
\captionof{table}{ Invariants for the $m-$list of systems of difference equations. For sake of brevity  we denote $H(u_2,s_2,v_1,t_1)$ as $\widehat H,$   $H(u,s,v,t)$ as $H$ and similarly for  $J$.} \label{table13}
\begin{tabular}{l|l|l}
  \hline
   System & Invariants $H, J$  & Invariance condition \\ \hline
  $m\mathcal{Q}4$           & $
  \begin{array}{l}
  H={\dis \frac{s(qu-pv)+t(qv-pu)-a(pqt-u-stv+pqsuv)}{s(qv-pu)+t(qu-pv)-a(pqs-v-stu+pqtuv)}},\\ [3mm]
  J={\dis \frac{1}{AB}\frac{(s+t) (u+v) \left((p+q) \left(\frac{p q (u v+1)}{u+v}-\frac{s t+1}{s+t}\right)+p Q+P q\right)^2}{(p-q) \left(1-p^2 q^2\right)+q (P+Q) \left(\frac{s
   t+1}{s+t}-\frac{p q (u v+1)}{u+v}\right)-Q (p Q+P q)}},\\ [3mm]
  \mbox{where} \; A:=\sqrt{p^2 s - u - s u (s - p^2 u) + 2 P s u},\\ [3mm]
  \;\;\;\;\;\; B:=\sqrt{q^2 t - v - t v (t - q^2 v) + 2 Q t v},\\ [3mm]
  \;\;\;\;\;\; a:={\dis\frac{pQ-qP}{1-p^2q^2}} \\ [3mm]
  \end{array}
  $  & $  \widehat H={\dis \frac{1}{H}},\;\widehat J=J$ \\ [3mm]  \hline
  $m\mathcal{Q}3^{\delta}$  & $ \begin{array}{l}
  H={\dis \frac{-c(u+vst)+a(tu+sv)-b(su+tv)-dt}{-c(v+ust)+a(su+tv)-b(tu+sv)-ds}}\\ [3mm]
  J={\dis \frac{\Pi(s+t)-pq(u+v)\left((s+t)(1+pq)-(p+q)(1+st)\right)}{\sqrt{\delta (1-p^2)^2s+pu(p-s)(1-ps)}\sqrt{\delta (1-q^2)^2t+qv(q-t)(1-qt)}}},\\ [3mm]
  \mbox{where}\;a:=p-1/p,\;\;b:=q-1/q,\;\;c:=p/q-q/p,\;\;d:=\delta a b c\\ [3mm]
  \;\;\;\;\;\;\Pi:=\delta (p+q)(1-p^2)(1-q^2)\\ [3mm]
  \end{array}$  & $\widehat H={\dis \frac{1}{H}},\;\widehat J=J$ \\ [3mm] \hline
  $m\mathcal{Q}1^{\delta}$  & $ \begin{array}{l}
  H={\dis \frac{\delta pq(p-q)t+p(t-1)(u-sv)-q(s-1)(u-tv)}{\delta pq(p-q)s+p(t-1)(v-su)-q(s-1)(v-tu)}}, \\ [3mm]
   J={\dis\frac{-\delta pq(s+t)+(1-s)(1-t)(u+v)}{\sqrt{\delta p^2s-u(1-s)^2}\sqrt{\delta q^2t-v(1-t)^2}}} \\ [3mm]
   \end{array} $  & $\widehat H={\dis \frac{1}{H}},\;\widehat J=J$ \\ [3mm] \hline
  $m\mathcal{A}2$           & $H={\dis\frac{(1+pq)(u+v)-(p+q)(1+uv)}{(p-q)(1-uv)-(1-pq)(u-v)}},\;\; J={\dis\frac{s}{t}} $  & $\widehat H=-H,\;\widehat J=J $ \\ [3mm] \hline
  $m\mathcal{A}1^{\delta}$  & $\begin{array}{l}
  H={\dis \frac{-\delta pq(p-q)t+p(t+1)(u+sv)-q(s+1)(u+tv)}{-\delta pq(p-q)s+p(t+1)(v+su)-q(s+1)(v+tu)}},\\ [3mm]
  J={\dis\frac{-\delta pq(s+t)+(1+s)(1+t)(u+v)}{\sqrt{\delta p^2s-u(1+s)^2}\sqrt{\delta q^2t-v(1+t)^2}}} \\ [3mm]
  \end{array} $  & $\widehat H={\dis \frac{1}{H}},\;\widehat J=J $ \\ [3mm] \hline
  $m\mathcal{H}3^{\delta}$  & $H={\dis\frac{\delta p+u}{\delta q+v}},\;\; J={\dis\frac{s}{t}} $  & $ \widehat H={\dis \frac{1}{H}},\;\widehat J=J$ \\ [3mm] \hline
  $m\mathcal{H}1$           & $H={\dis u-v+\frac{p-q}{2}},\;\; J={\dis\frac{s}{t}} $  & $\widehat H=-H,\;\widehat J=J $ \\ [3mm]
  \hline
\end{tabular}
\end{table}

\begin{table}[!h]
\captionof{table}{ Invariants for the $a-$list of systems of difference equations. For sake of brevity  we denote $H(u_2,s_2,v_1,t_1)$ as $\widehat H,$   $H(u,s,v,t)$ as $H$ and similarly for  $J$.} \label{table14}
\begin{tabular}{l|l|l}
  \hline
   Bond system & Invariant $H, \; J$ &  Invariance condition \\ \hline
  $a\mathcal{Q}3^{\delta}$       & $
                           \begin{array}{l}
                           H={\dis \frac{2abc\delta+(a-b+c)st+(-a+b+c)uv}{(a+b+c)tu+(a+b-c)sv-2bsu-2atv}},\\ [3mm]
                           J= {\dis\frac{\Pi+pq\left((1+p)(1+q)st+(-1+p+q-pq)uv\right)}{AB}}, \\ [3mm]
                          \mbox{where}\; a:=p-1/p,\;\;b:=q-1/q,\;\;c:=p/q-q/p,\\ [3mm]
                          \;\;\;\;\;\;A:= \sqrt{4\delta(1-p^2)^2+p(1+p)^2s^2-p(1-p)^2u^2},\\ [3mm]
                          \;\;\;\;\;\;B:= \sqrt{4\delta(1-q^2)^2+q(1+q)^2t^2-q(1-q)^2v^2},\\ [3mm]
                          \;\;\;\;\;\;\Pi:=2\delta (p+q)(1-p^2)(1-q^2) \\ [3mm]
                            \end{array}$  & $\widehat H=-H,\;\widehat J=J$ \\ [3mm] \hline
  $a\mathcal{Q}2$                & $ \begin{array}{l}
   H={\dis\frac{p^2-pq+q^2-u-v+\frac{st}{pq}}{(p-q)(s-t)+(u-v)\frac{pt-qs}{pq}}},\\ [3mm]
   J={\dis \frac{pq(p^2-pq+q^2)+st-pq(u+v)}{\sqrt{p^4+s^2-2p^2u}\sqrt{q^4+t^2-2q^2v}}} \\ [3mm]
   \end{array}$  & $\widehat H=-H,\;\widehat J=J$ \\[3mm] \hline
  $a\mathcal{Q}1^{\delta}$           & $ H=u-v,$ \;\; $J={\dis\frac{qs-pt}{\delta pq-st}} $ & $\widehat H=-H,\;\widehat J=J$ \\ [3mm] \hline
  $a\mathcal{A}1^{\delta}$       & $ H={\dis \frac{pv-qu}{\delta^2pq-uv}},$ \;\; $ J=s-t$ & $\widehat H=-H,\;\widehat J=J$ \\[3mm] \hline
  $a\mathcal{H}3^{\delta}$       & $ \begin{array}{l}
  H={\dis \frac{2\delta (p+q)+uv-st}{(p+q)(tu+sv)-2qsu-2ptv}},\\ [3mm]
  J={\dis\frac{2\delta (p+q)+uv-st}{\sqrt{4\delta p+u^2-s^2}\sqrt{4\delta q+v^2-t^2}}} \\ [3mm]
  \end{array}$  & $\widehat H=-{\dis \frac{1}{H}},\;\widehat J=J$ \\ [3mm] \hline
  $a\mathcal{H}2$                & $H={\dis\frac{u+p}{v+q}}, $ \;\; $J=s-t $ & $\widehat H={\dis \frac{1}{H}},\;\widehat J=J$ \\  [3mm] \hline
  $a\mathcal{H}1$                & $H=u-v, $\;\;  $J=s-t $ & $\widehat H=-H,\;\widehat J=J$ \\
  \hline
\end{tabular}
\end{table}

\subsection{Vertex systems}
\begin{theorem} \label{theorem1}
The corresponding  systems of vertex equations to the bond systems $m\mathcal{A}2, m\mathcal{H}3^{\delta}, m\mathcal{H}1,$ $a\mathcal{Q}1^{\delta},$ $a\mathcal{A}1^{\delta},$ $a\mathcal{H}2$ and $a\mathcal{H}1,$  are point equivalent to one of the following  systems
\begin{align*}
&\begin{aligned}
\tanh (w_{12}-w)  \tanh (z_{1}-z_{2}) +\frac{p-q}{p+q}=0,
&&
\tanh (z_{12}-z)  \tanh (w_{1}-w_{2}) +\frac{p-q}{p+q}=0,
\end{aligned} & (double-H3^{0})\\
&\begin{aligned}
(x_{12}-x)  \tanh (z_{1}-z_{2}) =p-q,
&&
(x_{1}-x_{2}) \tanh (z_{12}-z)   =p-q,
\end{aligned} & (double-H2)\\
&\begin{aligned}
(x_{12}-x)(y_1-y_2)=p-q, && (y_{12}-y)(x_1-x_2)=p-q,
\end{aligned} & (double-H1)
\end{align*}
\end{theorem}
\begin{proof}
First, we prove that the vertex system of equations associated with $m\mathcal{A}2$ leads to $double-H3^{0}.$ From Table \ref{table13}, we have that $a\mathcal{Q}1^{\delta}$ for $n=2$ satisfies the following invariant conditions:
$$
\frac{s_2}{t_1}=\frac{s}{t},\quad \frac{(1+pq)(u_2+v_1)-(p+q)(1+u_2v_1)}{(p-q)(1-u_2v_1)-(1-pq)(u_2-v_1)}+\frac{(1+pq)(u+v)-(p+q)(1+uv)}{(p-q)(1-uv)-(1-pq)(u-v)}=0
$$
The first condition from above guarantees the existence of a potential function $x$ such that
$$
s=\frac{x_1}{x},\quad t=\frac{x_2}{x},
$$
whereas the second condition,  guarantees the existence of a potential function $w$ such that
$$
u=\frac{p^2-1+(p^2+1)\tanh(w_1+w)}{2p\tanh(w_1+w)},\quad v=\frac{pq-1+(pq+1)\tanh(w_2+w)}{p-q+(p+q)\tanh(w_2+w)}.
$$
The bond system $m\mathcal{A}2$ in terms of the potential functions $x$ and $w$ reads
\begin{equation} \label{dh30}
\frac{x_1-x_2}{x_1+x_2}\tanh(w_{12}-w)+\frac{p-q}{p+q}=0,\quad \frac{x_{12}-x}{x_{12}+x}\tanh(w_1-w_2)+\frac{p-q}{p+q}=0,
\end{equation}
and after the substitution $x=e^{2z},$ becomes exactly  $double-H3^{0}$.
Second, we prove that the vertex system of equations associated with $a\mathcal{Q}1^{\delta}$ is exactly $double-H2.$   From Table \ref{table14}, we have that $a\mathcal{Q}1^{\delta}$ for $n=2$ satisfies the following invariant conditions:
$$
u_2-v_1=v-u,\quad \frac{qs_2-pt_1}{\delta^2pq-s_2t_1}=\frac{qs-pt}{\delta^2pq-st}.
$$
The first condition from above guarantees the existence of a potential function $x$ such that
$$
u=x_1+x,\quad v=x_2+x,
$$
whereas the second condition, for $\delta=1$, guarantees the existence of a potential function $z$ such that
$$
s=p\tanh(z_1-z),\quad t=q\tanh(z_2-z).
$$
The difference system $a\mathcal{Q}1^{\delta}$ for $n=2$ in terms of the potential functions $x$ and $z$ is exactly $double-H2.$

Third, the proof that the bond system $a\mathcal{H}1$ leads exactly to $double-H1,$ is straight forward.

The remaining bond systems $ m\mathcal{H}3^{\delta}, m\mathcal{H}1,$  $a\mathcal{A}1^{\delta},$ and $a\mathcal{H}2,$ in similar manner lead to  vertex systems. Namely, $ m\mathcal{H}3^{\delta}$ leads to $double-H3^0,$ $m\mathcal{H}1$ leads to $double-H1,$ $a\mathcal{A}1^{\delta}$ leads to $double-H2$ and finally $a\mathcal{H}2,$ leads to $double-H2.$

\end{proof}

\begin{remark}
The two-component lattice equation $double-H1$  alternatively can be derived from the matrix version of $H1$ \cite{Field:2005}, when $2\times 2$ off-diagonal matrices are used. A Lax representation for $double-H1,$  was  considered in \cite{Bridgman:2013}.
\end{remark}
\begin{remark}
The two-component lattice equation $double-H3^{0}$ equivalently reads
$$
p\cosh(w_{12}-w+z_1-z_2)-q\cosh(z_1-z_2+w-w_{12})=0,\quad p\cosh(z-z_{12}+w_2-w_1)-q\cosh(w_1-w_2+z-z_{12})=0,
$$
that  serves as a two-component generalization of the Hirota's discrete sine-Gordon equation \cite{Hir-sG}.
\end{remark}

\noindent {\bf Rational form of the double $H$-type equations.}
The system of equations $double-H3^0$ can be written in rational form if we set $x=e^{2w}$, $y=e^{2z},$
$$
p (x_1 y_{12} + x_2 y) - q (y x_1 + y_{12} x_2) = 0,\quad
 p (y_2 x + y_1 x_{12}) - q (y_1 x + y_2 x_{12}) = 0, \quad (double-H3^0).
$$
If we set $y=e^{2z},$  $double-H2$ takes the rational form
$$
(y_1-y_2)(x-x_{12})+(p-q)(y_1+y_2) = 0,\quad
(y-y_{12})(x_1-x_2)+(p-q)(y+y_{12})= 0, \quad (double-H2).
$$
Finally, $double-H1$ is already in rational form, for  consistent presentation we present it again 
$$
(x_{12}-x)(y_1-y_2)=p-q, \quad (y_{12}-y)(x_1-x_2)=p-q. \quad
  (double-H1).
$$
\noindent {\bf Non-potential versions of the double $H$-type equations.}
 For a comprehensive study of non-potential versions of integrable lattice equations on the $\mathbb{Z}^2$ and the $\mathbb{Z}^3$ lattice, see \cite{KaNie:2018}. Here we present the non-potential version of the double $H$-type equations which were presented earlier.
\begin{itemize}
\item  By performing the difference substitution  $h=e^{2 (w_1-w_2)}$, $g=e^{2 (z_1-z_2)},$ the $double-H3^0$ reads
\begin{equation}\label{nonpot1}
\begin{array}{ll}
\displaystyle
\frac{h_{12}}{h}=\frac{(qg_1-p)}{(pg_1-q)}  \frac{(pg_2-q)}{(qg_2-p)},&  \quad
\displaystyle
\frac{g_{12}}{g}=\frac{(qh_1-p)}{(ph_1-q)}  \frac{(ph_2-q)}{(qh_2-p)}.
\end{array}
\end{equation}
The system of equations (\ref{nonpot1}), is regarded as the non-potential version of $double-H3^0$. Note that under the reduction $h=g,$ (\ref{nonpot1}) reduces to the lattice sine-Gordon equation \cite{Volkov-1992,Bobenko-1993}.
\item  Under the  difference substitution  $h=x_1-x_2$, $g=e^{2 (z_1-z_2)},$ $double-H2$ reads
\begin{equation} \label{nonpot2}
\begin{array}{ll}
\displaystyle
h_{12}-h=\frac{2(p-q)}{g_1-1}-\frac{2(p-q)}{g_2-1},& \quad
\displaystyle
\frac{g_{12}}{g}=\frac{(h_1+p-q)}{(h_1-p+q)}  \frac{(h_2-p+q)}{(h_2+p-q)}.
\end{array}
\end{equation}
The system of equations (\ref{nonpot2}), is regarded as the non-potential version of $double-H2$.
\item The non-potential version of   $double-H1,$ can be obtained through the difference substitution $h=x_1-x_2$, $g=y_1-y_2.$ It reads
\begin{equation} \label{nonpot3}
\begin{array}{ll}
\displaystyle
h_{12}-h=\frac{p-q}{g_1}-\frac{p-q}{g_2}, & \qquad
\displaystyle
g_{12}-g=\frac{p-q}{h_1}-\frac{p-q}{h_2}.
\end{array}
\end{equation}
The system of equations (\ref{nonpot3}), under the reduction $h=g,$  reduces to the lattice Hirota's KdV equation \cite{hirota-0}.
\end{itemize}
As a final remark,  equation (\ref{nonpot1}) constitute of two coupled copies of the lattice sine-Gordon equation. Equation (\ref{nonpot3}) constitute of two coupled copies of the lattice Hirota's KdV equation. Finally, equation (\ref{nonpot2}) constitute of a copy of the lattice sine-Gordon equation coupled with a copy of the lattice Hirota's KdV equation and as it stands it seems that it does not possesses a scalar analogue.\\
\noindent {\bf 5-point schemes associated with  the double $H$-type equations.}
The system of vertex equations $double-H1$ serves as parameterless (auto)-B\"acklund transformation 
for a 5-point equation. Namely, both $x$ and $y$ obey the equation
\begin{equation}
\label{5p1}
\begin{array}{ll}
\displaystyle
\frac{p-q_{-2}}{y_{-21}-y}+\frac{p_{-1}-q}{y_{-12}-y}=\frac{p-q}{y_{12}-y}+\frac{p_{-1}-q_{-2}}{y_{-1-2}-y}.
\end{array}
\end{equation}

Similiary, the system of vertex equations $double-H3$ serves as  (auto)-B\"acklund transformation for a 5-point equation i.e.
 $x=e^{2w}$ and $y=e^{2z}$ obeys the equation
\begin{equation}
\label{5p3}
\begin{array}{ll}
\displaystyle
\frac{(q y_{12}-p y)}{(p y_{12}-q y)} \frac{(p_{-1} y-q_{-1} y_{-1-2})}{(q_{-2} y-p_{-1} y_{-1-2})}
=
\frac{(q_{-2} y_{1-2}-p y)}{(p y_{1-2}-q_{-2} y)}
\frac{(p_{-1} y-q y_{-12})}{(q y-p_{-1} y_{-12})}  .
\end{array}
\end{equation}

Finally, the system of vertex equations $double-H2$ serves as (non auto)-B\"acklund transformation between 5-point equation (\ref{5p1})
where $y=e^{2z}$ and
\begin{equation}
\begin{array}{ll}
\displaystyle
\frac{(x_{12}-x+p-q)}{(x_{12}-x-p+q)} \frac{(x_{-1-2}-x+p_{-1}-q_{-2})}{(x_{-1-2}-x-p_{-1}+q_{-2})}=
\frac{(x_{1-2}-x+p-q_{-2})}{(x_{1-2}-x-p+q_{-2})}\frac{(x_{-12}-x+p_{-1}-q)}{(x_{-12}-x-p_{-1}+q)}.
\end{array}
\end{equation}

\section{Conclusions} \label{Section:7}
In this work we have presented two lists, the $m-$list and the $a-$list of  multi-component systems of integrable equations in bond variables. These multi-component systems are consistent with appropriate constraints imposed on the $\mathbb{Z}^2$ graph. Using these constraints we obtained known integrable quad-equations equations, as well as the discrete Krichever-Novikov equation, as  reductions of the  multi-component systems in bond variables.

The integrability of the members of the $m-$list and the $a-$list follows from the multi-dimensional compatibility of these members. Moreover, the multi-dimensional compatibility guarantees the existence of Lax pairs, as well as the existence of  Yang-Baxter maps associated with the members of both lists. 

Finally, from some members of the $m-$list and the $a-$list, we have obtained the corresponding vertex systems of equations, together with their rational form and their non-potential counterparts, as well as some lattice equations defined on $5-$point stencils.

Regarding all the members of the $m-$list and $a-$list, some open questions  that can be addressed in a future study include:
\begin{itemize}
\item The Liouville integrability as it was indicated in Proposition \ref{prop3.2}, for the motivating example.
\item The corresponding non-potential forms of the systems.
\item Entwining maps \cite{Kouloukas:2011,Kassotakis:2019,Rizos:2019} associated with the Yang-Baxter maps given in Appendix \ref{appA} and in Appendix \ref{appB}.
\item The derivation of the associated two-component vertex systems.
\end{itemize}

\appendix
\section{Yang-Baxter maps associated with the $m-$list of integrable difference systems in bond variables} \label{appA}
\begin{prop}\label{prop3.4}
The following families of maps $ R:(u,s,p;v,t,q)\mapsto (U,S,p;V,T,q),$    where
\begin{align*}
&\begin{aligned}
&U=\frac{at(pqs-v)+v(p-qst)}{a(s-pqv)+q-pst},&&
&S=\frac{av(1-pqtu)+t(qu-pv)}{a(pq-tu)+pu-qv},&\\
&V=\frac{a(pq-tu)+u(p-qst)}{as(1-pqtu)+q-pst},&&
&T=\frac{a(pqs-v)+s(pu-qv)}{au(s-pqv)+qu-pv},&\\
&\mbox{where} \;\; a:=\frac{pQ-qP}{1-p^2q^2},\;P^2:=p^4-\gamma p^2+1,\;Q^2:=q^4-\gamma q^2+1,&&
\end{aligned} & (m^c\mathcal{Q}4)\\
&\begin{aligned}
&U=\frac{(abc\delta+bv)st+v(ct-a)}{-b+s(at-c)},&&
&S=\frac{-btu+v(at-c)}{u(ct-a)+bv+abc\delta},&\\
&V=\frac{tu(bs+c)-au+abc\delta}{-b+s(at-c)},&&
&T=\frac{s(abc\delta-au)+v(bs+c)}{av-u(cs+b)},&\\
&\mbox{where} \;\; a:=p-1/p,\;\; b=q-1/q,\; c:=p/q-q/p,&&
\end{aligned} & (m^c\mathcal{Q}3^{\delta})\\
&\begin{aligned}
&U=\frac{-\delta (p-q)pqst+tv(p-q+qs)-pv}{(q-p)s-q+pst},&&
&S=\frac{(p-q)v+t(qu-pv)}{(p-q)(\delta pq-tu)+pu-qv},&\\
&V=\frac{(p-q)(-\delta pq+tu)+u(qst-p)}{(q-p)s-q+pst},&&
&T=\frac{(p-q)(-\delta pqs+v)-s(pu-qv)}{(q-p)su-qu+pv},&
\end{aligned} & (m^c\mathcal{Q}1^{\delta})\\
&\begin{aligned}
&U=\frac{bv+st(c-a v)}{a-bst-cv},&&
&S=t\frac{au-v(b+cu)}{c+bu-av},&\\
&V=\frac{c+u(b-ast)}{a-st(b+cu)},&&
&T=s\frac{c+bu-av}{au-v(b+cu)},&\\
&\mbox{where} \;\; a:=p-1/p,\;\; b=q-1/q,\; c:=p/q-q/p,&&
\end{aligned} & (m^c\mathcal{A}2)\\
&\begin{aligned}
&U=\frac{(p-q)(tv-\delta pqst)+v(p-qst)}{(q-p)s+q-pst},&&
&S=\frac{(p-q)v-t(qu-pv)}{(p-q)(\delta pq-tu)-pu+qv},&\\
&V=\frac{(p-q)(tu-\delta pq)+u(p-qst)}{(q-p)s+q-pst},&&
&T=\frac{(p-q)(v-\delta pqs)+s(pu-qv)}{(q-p)su+qu-pv},&
\end{aligned} & (m^c\mathcal{A}1^{\delta})\\
&\begin{aligned}
&U=\frac{\delta (p^2-q^2)st+v(p-qst)}{q-pst},&&
&S=t\frac{qu-pv}{\delta (p^2-q^2)+pu-qv},&\\
&V=\frac{\delta (p^2-q^2)+u(p-qst)}{q-pst},&&
&T=s\frac{\delta (p^2-q^2)+pu-qv}{qu-pv},&
\end{aligned} & (m^c\mathcal{H}3^{\delta})\\
&\begin{aligned}
&U=v+(p-q)\frac{st}{1-st},&&
&S=t\left(1-\frac{(p-q)}{p-q+u-v}\right),&\\
&V=u+(p-q)\frac{1}{1-st},&&
&T=s\left(1+\frac{(p-q)}{u-v}\right),&
\end{aligned} & (m^c\mathcal{H}1)
\end{align*}
\begin{itemize}
\item are quadrirational non-involutive Yang-Baxter  maps
\item each  map is associated with  the corresponding difference system of Proposition  \ref{prop3.1} with $n=2.$
\end{itemize}
\end{prop}

\newpage

\section{Yang-Baxter maps associated with the $a-$list of integrable difference systems in bond variables} \label{appB}
\begin{prop}\label{prop3.3}
The following families of maps $ R:(u,s,p;v,t,q)\mapsto (U,S,p;V,T,q),$    where\begin{align*}
&\begin{aligned}
&U=s-2\frac{-abc\delta+bs^2+(a+c)st-v(bs+at)}{2bs+(a+b+c)t+(a-b-c)v},&&
S=u-2\frac{-abc\delta+bu^2-(a+c)tu+v(at-bu)}{2bu+(-a+b-c)t+(-a-b+c)v},\\
&V=t-2\frac{-abc\delta+at^2-(b+c)tu+s(at-bu)}{2at+(a-b-c)u+(a+b-c)s},&&
T=v-2\frac{-abc\delta+av^2-(b+c)uv+s(bu-av)}{2av-(a+b+c)u+(-a+b-c)s},\\
&where \;\; a:=p-1/p,\;\; b=q-1/q,\; c:=p/q-q/p,
\end{aligned} & (a^c\mathcal{Q}3^{\delta})\\
&\begin{aligned}
&U=v+(p-q)\frac{pq(p^2-pq+q^2+s+t-2v)-st}{pq(p-q)+qs+pt},&&
S=t-(p-q)\frac{pq(p^2-pq+q^2-u-v)+t(t+u-v)}{(pq+t)(p-q)-q(u-v)},\\
&V=u-(p-q)\frac{pq(-p^2+pq-q^2+s+t+2u)+st}{pq(p-q)+qs+pt},&&
T=s+(p-q)\frac{pq(p^2-pq+q^2-u-v)+s(s+u-v)}{(pq-s)(p-q)-p(u-v)},
\end{aligned} & (a^c\mathcal{Q}2)\\
&\begin{aligned}
&U=v-(p-q)\frac{\delta pq+st}{qs+pt},&&
S=t-(p-q)\frac{-\delta pq+t(t+u-v)}{(p-q)t-q(u-v)},\\
&V=u-(p-q)\frac{\delta pq+st}{qs+pt},&&
T=s-(p-q)\frac{-\delta pq+s(s+u-v)}{(p-q)s+p(u-v)},
\end{aligned} & (a^c\mathcal{Q}1^{\delta})\\
&\begin{aligned}
&U=p\frac{\delta q(p-q)+v(s+t)}{(p-q)v+q(s+t)},&&
S=t+(p-q)\frac{\delta pq-uv}{qu-pv},\\
&V=q\frac{\delta p(p-q)+u(s+t)}{(p-q)u+p(s+t)},&&
T=s-(p-q)\frac{\delta pq-uv}{qu-pv},
\end{aligned} & (a^c\mathcal{A}1^{\delta})\\
&\begin{aligned}
&U=\frac{-2\delta (p^2-q^2)-(p-q)st+(p+q)sv+2ptv}{2qs+(p+q)t+(p-q)v},&&
S=\frac{2\delta (p^2-q^2)-(p+q)tu+(p-q)uv+2ptv}{-2qu+(p-q)t+(p+q)v},\\
&V=\frac{-2\delta (p^2-q^2)-(p-q)st+(p+q)tu+2qsu}{2pt+(p+q)s+(p-q)u},&&
T=\frac{2\delta (p^2-q^2)-(p+q)sv+(p-q)uv+2qsu}{-2pv+(p+q)u+(p-q)s},
\end{aligned} & (a^c\mathcal{H}3^{\delta})\\
&\begin{aligned}
&U=v-p+q-2(p-q)\frac{q+v}{p-q-s-t},&&
S=t+(p-q)\frac{p+q+u+v}{p-q+u-v},\\
&V=u+p-q-2(p-q)\frac{p+u}{p-q-s-t},&&
T=s-(p-q)\frac{p+q+u+v}{p-q+u-v},
\end{aligned} & (a^c\mathcal{H}2)\\
&\begin{aligned}
&U=v-\frac{p-q}{s+t},&&
S=t-\frac{p-q}{u-v},\\
&V=u-\frac{p-q}{s+t},&&
T=s+\frac{p-q}{u-v},
\end{aligned} & (a^c\mathcal{H}1)
\end{align*}
\begin{itemize}
\item are quadrirational non-involutive Yang-Baxter  maps
\item each  map is associated with  the corresponding difference system of Proposition  \ref{prop3.2} with $n=2.$
\end{itemize}
\end{prop}

\section{Lax pairs of Yang-Baxter maps} \label{appC}

Here it is convenient to use the following notation. We denote the set of variables  $(u,s,p;v,t,q)$ and $(U,S,p;V,T,q)$, where a Yang-Baxter map $R$ acts, respectively as $(x_i,y_i,\alpha_i;x_j,y_j,\alpha_j)$ and $(x'_i,y_i',\alpha_i;x'_j,y'_j,\alpha_j),$ for some fixed $i\neq j\in \mathbb{N}.$ Note  that different subscripts denote different variables and not shifts as it was considered until now in the manuscript.

Now we are ready to recall the following definition.

\begin{dfn}
A map $R: (u,s,p;v,t,q)\mapsto(U,S,p;V,T,q)$ is a Yang-Baxter map if the following relation holds
$$
R_{12}R_{13}R_{23}=R_{23}R_{13}R_{12},
$$
regarded as a composition of maps on a auxiliary space with an additional third dimension indicated by a subscript $k$, i.e.
$$
R_{23}( x_i , y_i , a_i ;    x_j , y_j   ,  a_j  ;   x_k  , y_k ,  a_k ) =
               ( x_i , y_i , a_i ; x'_j , y'_j  ,  a_j  ;   x'_k , y'_k , a_k )  \,.
$$
etc.
\end{dfn}
The connection of Yang-Baxter maps with integrable quad-equations was originated in \cite{Tasos,Papageorgiou_2007} and complemented in \cite{KaNie,KaNie1,KaNie3,KaNie:2018,Kouloukas:2012} where also the connection with higher degree  quad-relations was established. Moreover, the interplay between Yang-Baxter maps and discrete integrable systems led to fruitful  results \cite{atk-2013,BAZHANOV2018509,Caudrelier_2013,Atkinson:2018,Dimakis2018,Dimakis2018ii,AtkNie,Rizos:2013,Kassotakis:2013,Grahovski:2016,Mikhailov2016,Kouloukas:2018,Kassotakis:2019,Rizos:2019b}.
The notion of Lax matrices for Yang-Baxter maps was introduced in \cite{Veselov:2003b} and is in accordance for Lax pairs of bond systems given in Definition \ref{dfn111}.

\begin{dfn}
(i) $L(x,y;\alpha,\lambda)$ is called a Lax matrix of the Yang-Baxter map
$R: (x_1,y_1,\alpha_1;x_2,y_2,\alpha_2)\mapsto (x_1',y_1',\alpha_1;x_2',y_2',\alpha_2)$, if the relation
$$R((x_1,y_1),(x_2,y_2))=((x_1',y_1'),(x_2',y_2'))$$
 implies that
\begin{equation}
L({x_2,y_2};\alpha_2,\lambda)\, L({x_1,y_1};\alpha_1,\lambda) =
L(x'_1,y_1;\alpha_1,\lambda)\, L({x'_2,y'_2};\alpha_2,\lambda)\,,
\label{eq:2-fact}
\end{equation}
for all $\lambda$.
$L(x,y;\alpha,\lambda)$ is called a strong Lax matrix of
$R$, if the converse also holds. \\
(ii) $L(x,y;\alpha,\lambda)$ satisfies the $n$-factorization property if
the identity
\begin{equation}
L(x'_n,y'_n;\alpha_n,{\lambda})\,\cdots
L(x'_2,y'_2;\alpha_{2},{\lambda})\,
L(x'_1,y'_1;\alpha_1,{\lambda}) \equiv
L(x_n,y_n;\alpha_n,{\lambda})\, \cdots
L(x_2,y_2;\alpha_{2},{\lambda})\,
L(x_1,y_1;\alpha_1,{\lambda})\,, \label{eq:3-fact}
\end{equation}
implies that ${{x}_i}' = {x}_i$, $y'_i=y_i$, $i=1,\ldots ,n$.
\end{dfn}
For all the Yang-Baxter maps presented in Appendices \ref{appA} and \ref{appB}, their associated Lax pairs can be obtained through the algorithmic procedure presented in \cite{Veselov:2003b}.
\begin{remark} \label{rem:nfact}
The $2$-factorization property of the Lax matrix $L$ corresponds to
the unitarity property of $R$, ($R_{21} \, R_{12} = \mathrm{Id}$) while the $3$-factorization property to the Yang-Baxter property.
\end{remark}

\begin{example}
The Yang-Baxter map $m^c\mathcal{H}1$ admits a strong Lax pair with the following Lax matrix
\begin{equation} \label{laxYB_H3xH1}
L(x,y;\alpha,\lambda) :=
\begin{pmatrix}
    0     & 0              & x\,y   & \lambda -\alpha - x \\
    0     & 0              & y      & -1             \\
    y     & y(\lambda-\alpha-x) & 0      &  0             \\
    1     & -x             & 0      &  0
\end{pmatrix}
\end{equation}
\end{example}

\begin{remark}
For the above Lax matrix \eqref{laxYB_H3xH1} the $n$-factorization property can be proved as follows.
The (right) null space of the linear transformation
in the LHS of (\ref{eq:3-fact}), for $\lambda=\alpha_1$, is two dimensional and spanned by the vectors
$$ n_1 = \left(x'_1,1,0,0\right)^t \,, \qquad n_2 = \left(0,0,1,y'_1\right)^t$$
Similarly, $\left(x_1,1,0,0\right)^t$ and $\left(0,0,1,y_1\right)^t$
span the (right) null space of the RHS linear transformation.
Because of the identity (\ref{eq:3-fact}), we conclude that $x_1={x'_1}$, $y_1={y'_1}$,
and the rightmost matrices cancel out.
Therefore, by induction, $L$ satisfies the $n$-factorization property.
\end{remark}
\begin{remark}
The $n$-factorization property for a Yang-Baxter map can be used to prove multi-dimensional compatibility for the corresponding bond system. So an alternative way to prove the 3D-compatibility of the bond systems of the $m-$list and the $a-$list, is to prove the $n$-factorization property for the associated to the bond systems Yang-Baxter maps.
\end{remark}


\end{document}